\documentclass[aps,prx,reprint,superscriptaddress,nofootinbib,floatfix]{revtex4-2}

\usepackage[utf8]{inputenc}

\usepackage{amsmath}
\usepackage{amssymb}
\usepackage{mathtools}
\usepackage{amsthm}
\usepackage{physics}
\usepackage{bbold}

\usepackage{graphicx}
\usepackage{xcolor}
\usepackage{booktabs}
\usepackage[caption=false]{subfig}

\usepackage{tikz}
\usetikzlibrary{quantikz2}
\usetikzlibrary{arrows.meta,positioning,calc,fit,backgrounds,decorations.pathmorphing}

\usepackage{thm-restate}

\usepackage[colorlinks=true, citecolor=red, linkcolor=red, urlcolor=red]{hyperref}

\usepackage{cleveref}

\newcommand{\KL}{D_{\mathrm{KL}}}
\newcommand{\TVD}{\mathrm{TVD}}

\newtheorem{definition}{Definition}

\newtheorem{corollary}{Corollary}

\begin{document}

\title{Quantum Fourier Generative Models Trainable at Large Scale}

\author{Cenk T\"uys\"uz}
\email{cenk.tuysuz@cern.ch}
\affiliation{European Organisation for Nuclear Research (CERN), 1211 Geneva, Switzerland}
\author{Oleksandr Kyriienko}
\affiliation{School of Mathematical and Physical Sciences, University of Sheffield, Sheffield S10 2TN, United Kingdom}
\author{Michele Grossi}
\affiliation{European Organisation for Nuclear Research (CERN), 1211 Geneva, Switzerland}

\date{\today}

\begin{abstract}
We propose an algorithmic framework for building and training quantum generative models corresponding to multivariate probability distributions.
Our model uses parallel Fourier feature maps for embedding continuous-valued variables combined with a forrelation-type quantum circuit for tuning Fourier coefficients of the quantum model.
Crucially, we develop a distinct training strategy where training is enabled at large scale by log-likelihood loss with unbiased Monte Carlo estimator based on Parseval's identity.
Unlike prior work that relied on maximal mean discrepancy (MMD) loss, our approach goes beyond matching just low frequency moments, while enabling efficient classical training.
Once the model is trained, we use inverse quantum Fourier transforms to map it into a separate sampling circuit in the computational basis.
We demonstrate the efficiency of the suggested framework by validating loss estimation at the scale of over 1000 qubits on a single GPU.
We show that univariate and bivariate models with highly non-trivial structure can be trained to low total variation distance, while fine-tuned IQP models with MMD loss show poor performance.
Comparing to classical baselines represented by normalizing flow and diffusion models, we show that our approach avoids oversmoothing and preserves multi-modal structure of the target.
Finally, we have deployed the trained models on superconducting quantum devices, successfully sampling distributions with per-sample execution times of approximately $300\,\mu\mathrm{s}$.
Our work shows that quantum generative models with the train-on-classical deploy-on-quantum approach can provide both high-quality structure at increased scale and fast sampling access needed for inference.
\end{abstract}

\maketitle

\section{Introduction}
\label{sec:intro}

\begin{figure*}[!t]
    \centering
    \includegraphics[width=\linewidth]{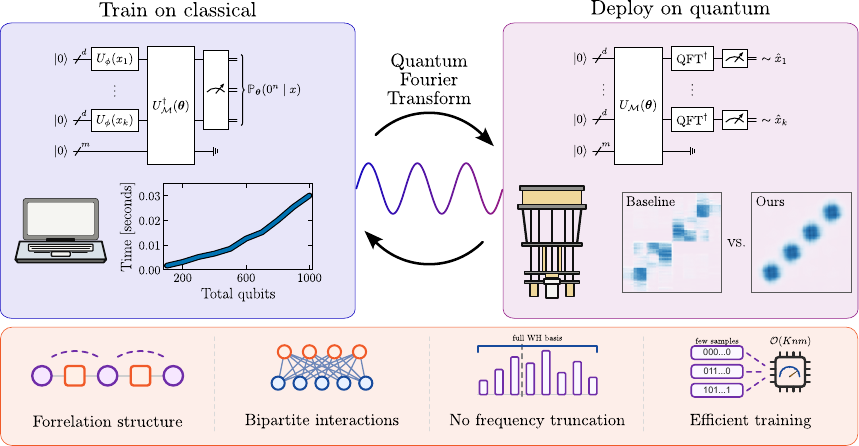}
    \caption{
    \textbf{Train-on-classical deploy-on-quantum workflow.}
    During training, data are encoded through a Fourier feature map and the model parameters are optimized classically in the latent Fourier representation.
    After training, the same parameters are deployed in a quantum sampling circuit, where inverse quantum Fourier transforms ($\mathrm{QFT}^\dagger$) produce samples in the computational basis.
    The construction combines a forrelation-inspired architecture, bipartite visible-hidden interactions, and efficient classical training. The training circuit computes Eq.~\eqref{eq:marginal_single}, while the sampling circuit computes Eq.~\eqref{eq:sampling}.
    The training loss is built from visible marginal likelihoods, making it more sensitive to a target distribution than MMD-type losses.
    }
    \label{fig:workflow}
\end{figure*}

Generative modelling is based on representing, learning, and sampling complex multivariate probability distributions \cite{bengio2013,rezende15,song2021denoising,song2021scorebased}.
Quantum generative modeling takes advantage of the intrinsic probabilistic nature of quantum systems and the ability to sample through projective measurements \cite{Liu2018,Zoufal2019,Coyle2020}.
The exponential sampling advantage for certain families of quantum circuits \cite{bremner2010,boson-sampling-adv,Arute2019,RevModPhys.95.035001} motivated the development of trainable quantum generative models (QGMs) that can learn efficiently while improving the scaling for sample generation~\cite{huang2025}. 
Several QGM families have been proposed, including quantum circuit Born machines (QCBMs) \cite{Liu2018,Zoufal2019,Coyle2020}, quantum Boltzmann machines (QBMs) \cite{amin2018, coopmans2024, tuysuz2024, demidik2025, demidik2026}, physics-informed samplers \cite{kyriienko2024,Paine2023,Bak02026}, hybrid and photonic architectures \cite{Oszmaniec2022,Wang2023,kailasanathan2026,bacarreza2026} etc. 

Training and inference of quantum generative models crucially depends on their architecture and loss~\cite{herbst2025}. The trade-off between the expressivity and trainability of variational quantum circuits \cite{Holmes2022,CerezoHolmes2025} often limits training to smaller system sizes due to vanishing gradients~\cite{mcclean2018,Larocca2025} and local minima \cite{you2021,anschuetz2023}, given that efficient training implies low shot noise when training on quantum devices. However, this trade-off is alleviated if gradients can be calculated classically, at large scale and the required precision, avoiding costly parameter shift rules \cite{Schuld2019,Kyriienko2021GPSR,Wierichs2022}. Here, generative modelling offers a distinct setting, as proposed for QGMs based on instantaneous quantum polynomial (IQP) circuits~\cite{Kasture2023}. For this architecture training can be done on classically tractable loss, while inference is computationally hard and off-loaded to quantum devices. The separation of training and inference stages hence led to emergence of QGMs based on \emph{train-on-classical deploy-on-quantum} approach~\cite{Recio-Armengol2025, kurkin2025}, which also extended to fermionic linear optics circuits~\cite{bako2025}, shallow circuits \cite{huang2025}, and boson-sampling architectures~\cite{Kolarovszki2026,gottlieb2026,kurkin2026}.

For IQP-based QCBMs, the classically-tractable training objective is a maximum mean discrepancy (MMD) loss, computed from expectation values of diagonal observables~\cite{vandennest2010, rudolph2024}. Effectively, this relies on matching lowest moments of the target distribution performed in a spectral domain \cite{belis2026}, where in a Walsh-Hadamard representation the kernel puts large weight on low-order correlations~\cite{herrerogonzalez2025, shen2026}. When the target has strong correlations across features, the relevant structure can lie in high-order components that the kernel weights weakly, so these correlations contribute little to the loss and its gradient and MMD training can miss them even when the model class can represent the target, as we show in Section~\ref{sec:numerics}. The underlying issue is not the use of spectral information itself~\cite{wakeham2024,belis2026}, but the fixed spectral weighting imposed by a prescribed MMD kernel. To get high-quality matching of QGMs to the target distribution, we need qualitatively new approaches to model architecture building and training objectives.

In this work, we propose a different algorithmic framework for building generative models based on \emph{train-on-classical deploy-on-quantum} approach (see Fig.~\ref{fig:workflow} for a visual summary). Our framework combines: 1) working in Fourier space; 2) forrelation-type circuits with visible-hidden qubit separation; 3) efficient Monte Carlo (MC) estimator of a log-likelihood loss, advancing beyond MMD and trained at large scale; 4) quantum Fourier transformed output for fast sampling over extended support. We implement this workflow \emph{in-silico}, evaluating up to 1000 qubits and testing for different targets, as well as running inference on real quantum hardware. The proposed framework allows us to model multivariate distributions of continuous variables relevant for correlated scientific data as well as financial analysis and weather forecasting. One domain where such models are important is high-energy physics~\cite{hep2024}, given complexity and large dimensionality of distributions. 

\section{Framework}

The proposed algorithmic framework combines specific choices for embedding, measurement adaptation circuit, basis transform for the sampling stage, and a bespoke training procedure. We are detailing these ingredients in the following subsections.

\subsection{Model design}

At the circuit level, our construction combines elements of the differentiable quantum generative model (DQGM)~\cite{kyriienko2024}, forrelation circuits~\cite{Aaronson2018,Kasture2023,Umeano2026forrelation}, and circuit families with classically tractable quantum amplitudes~\cite{bravyi2021}. A Fourier feature map $U_{\phi}(x)$ (see Eq.~\eqref{eq:encoding}) encodes the data using a Hadamard layer followed by a diagonal unitary~\cite{kyriienko2024}, which mirrors the structure of forrelation circuits. This allows to establish a relationship between implicit and explicit probabilistic models with a fixed basis transformation, and query models during training at selected points $\mathcal{X}$.

We design the model with a bipartite architecture with a natural latent-variable interpretation. The $n$ visible qubits encode the data, and the $m$ hidden qubits mediate correlations between them through pairwise interactions, analogous to a restricted Boltzmann machine~\cite{Hinton2012}. The interactions remain diagonal, and the hidden qubits extend expressivity without introducing direct couplings among the visible qubits.

The bipartite restriction also determines the classical cost of the model. Forrelation-type amplitudes are hard to approximate in general, yet confining the interactions to visible-hidden pairs decouples the visible qubits once the hidden configuration is fixed. The amplitude then factorizes over the visible qubits, so at a fixed hidden configuration it can be evaluated classically in polynomial time, while the circuit retains a nontrivial structure.

The model unitary \( U_\mathcal{M}(\boldsymbol{\theta}) \) acts on \( n \) visible qubits and \( m \) hidden qubits, with \( N = n + m \) total qubits. It follows a forrelation-type structure with two diagonal layers separated by a fixed intermediate transformation,
\begin{equation*}
\begin{split}
U_\mathcal{M}(\boldsymbol{\theta}) = \bigl(I^{\otimes n} \otimes H^{\otimes m} \bigr) \, D(\boldsymbol{\theta}^{(2)})\,
    &\bigl(\mathrm{RY}(-\pi / 2)^{\otimes n} \otimes I^{\otimes m}\bigr) \\
&\times D(\boldsymbol{\theta}^{(1)})\, H^{\otimes N} ,
\end{split}
\end{equation*}
where \( \boldsymbol{\theta} = \{\boldsymbol{\theta}^{(1)}, \boldsymbol{\theta}^{(2)} \} \) denotes the set of trainable parameters. The diagonal layers are restricted to bipartite interactions between visible and hidden qubits.

\begin{definition}[Bipartite diagonal layer]
\label{def:diagonal}
The diagonal unitary can be written compactly using Pauli strings as
\begin{equation}
    D(\boldsymbol{\alpha}) 
    = \prod_{s \in \mathcal{S}} e^{-i \alpha_s Z(s)}, \nonumber
\end{equation}
where \(Z(s) = \bigotimes_{j=1}^N Z^{s_j}\). The set \(\mathcal{S}\) is restricted to strings of weight at most two, with the additional constraint that all weight-two terms connect a visible and a hidden qubit.
\end{definition}

The construction extends straightforwardly to multivariate inputs. For a $k$-dimensional input $\boldsymbol{x} = (x_1,\ldots,x_k)$, each feature $x_l$ is encoded on a dedicated block of $d$ visible qubits, yielding a total of $n = k d$ visible qubits. The hidden qubits couple to all feature blocks, allowing the model to capture cross-feature correlations. The encoding is therefore given by the tensor product \(\bigotimes_{l=1}^k U_\phi(x_l)\), while the model unitary acts jointly on all visible and hidden qubits.

The corresponding training and sampling circuits are illustrated in Fig.~\ref{fig:workflow}. In the sampling mode, the inverse QFT decomposes across feature blocks, so that a single computational basis measurement yields a joint sample \((\hat{x}_1,\ldots,\hat{x}_k)\).

The training signal is the marginal probability of the all-zero visible outcome, obtained by summing over the hidden register,
\begin{equation}
    \mathbb{P}_{\boldsymbol{\theta}}(0^n \mid x)
    = \sum_{\mathbf{b}_h \in \{0,1\}^m}
    \left| \bra{0^{n},\mathbf{b}_h}\, U_\mathcal{M}^\dagger(\boldsymbol{\theta})\, U_\phi(x)\, \ket{0^{N}} \right|^2 ,
\label{eq:marginal_single}
\end{equation}
where $U_\phi(x)$ encodes the input and the sum runs over all hidden strings $\mathbf{b}_h$. For each input $x$, the training circuit of Fig.~\ref{fig:workflow} returns this probability, which we take as the likelihood the model assigns to $x$. We fit the parameters by minimizing the empirical negative log-likelihood (NLL) over the training set $\mathcal{X}_\mathrm{train}$,
\begin{equation}
    \mathrm{NLL}(\boldsymbol{\theta}) = -\,\mathbb{E}_{x \in \mathcal{X}_\mathrm{train}}
    \left[ \log \mathbb{P}_{\boldsymbol{\theta}}(0^n \mid x) \right],
\label{eq:nll}
\end{equation}
which raises $\mathbb{P}_{\boldsymbol{\theta}}(0^n \mid x)$ on the training data and matches the model distribution to the target.

The cost of training lies entirely in evaluating the marginal $\mathbb{P}_{\boldsymbol{\theta}}(0^n \mid x)$. The sum in Eq.~\eqref{eq:marginal_single} runs over all $2^m$ hidden configurations, so a direct evaluation costs time exponential in $m$, while a MC estimate over the signed amplitudes suffers from the sign problem. In the next section we develop a classical estimator that evaluates this marginal in time linear in the number of visible qubits, combining the factorization over visible qubits with Parseval's identity.

We provide preliminaries in App.~\ref{app:background} and additional details regarding the model construction, intuition behind the choices and implementaion details in App.~\ref{app:model-details}.

\subsection{Classical training algorithm}
\label{sec:training}

Training requires the marginal $\mathbb{P}_{\boldsymbol{\theta}}(0^n\mid x)$ at every data point, which Eq.~\eqref{eq:marginal_single} writes as a sum over $2^m$ hidden configurations. A direct evaluation of this sum costs time exponential in $m$, and a Monte Carlo estimate over the signed amplitudes suffers from the sign problem. We remove both obstacles by reformulating the marginal as the average of a non-negative quantity over uniformly random hidden configurations. We first construct a function $g(\boldsymbol{\omega})$ of the hidden configuration that is evaluated in time $O(n\,d_{\max})$, and then show that the marginal equals the uniform average of $|g(\boldsymbol{\omega})|^2$.

We build $g$ from the amplitude ratio of the two diagonal layers of $U_\mathcal{M}$. Splitting $U_\mathcal{M}$ at its center produces two states whose basis amplitudes we denote $A(\mathbf{a},\boldsymbol{\omega})$ and $B(\mathbf{a},\boldsymbol{\omega})$, where $\mathbf{a}\in\{0,1\}^n$ and $\boldsymbol{\omega}\in\{0,1\}^m$ are the visible and hidden parts of a computational-basis string. The second layer rephases each basis state, so $|B|$ is constant and the ratio $\overline{A}/\overline{B}$ is well defined. Because $U_\mathcal{M}$ couples each visible qubit only to its hidden neighbours $\mathcal{N}(v)$ and never to another visible qubit, this ratio factorizes into per-qubit ratios $r_v(a_v;\boldsymbol{\omega})$, each depending on a single visible bit $a_v$ and on the hidden neighbours of $v$. These ratios obey $|r_v|\le\sqrt{2}$, which bounds the variance of the estimator below. We give the explicit local amplitudes in App.~\ref{app:derivation}.

\begin{restatable}[Bipartite factorization]{lemma}{factlemma}
\label{lem:factorization}
Let $g(\boldsymbol{\omega})$ be the uniform average of the amplitude ratio over the visible register. Since $U_\mathcal{M}$ has no visible-to-visible coupling, this average factorizes,
\begin{equation}
    g(\boldsymbol{\omega}) = \prod_{v=1}^{n} \frac{r_v(0;\boldsymbol{\omega}) + r_v(1;\boldsymbol{\omega})}{2},
    \label{eq:gfact}
\end{equation}
and $g(\boldsymbol{\omega})$ is evaluated in $O(n\,d_{\max})$, with $d_{\max}=\max_v|\mathcal{N}(v)|$.
\end{restatable}

Each factor in Eq.~\eqref{eq:gfact} is the average of $r_v$ over the two values of one visible bit, so the factorization replaces the average over the $2^n$ visible configurations by a product of $n$ local terms. Each term costs $O(d_{\max})$ through the product over the hidden neighbours of its visible qubit, so $g(\boldsymbol{\omega})$ is evaluated in $O(n\,d_{\max})$ with $d_{\max}\le m$. The factorization carries out the average over the visible register exactly, leaving only the hidden configuration $\boldsymbol{\omega}$ to be sampled.

\newpage
\begin{restatable}[Marginal estimator]{proposition}{estimatorprop}
\label{prop:estimator}
The training signal equals the uniform average of $|g|^2$ over the hidden register,
\begin{equation}
    \mathbb{P}_{\boldsymbol{\theta}}(0^n\mid x)
    = \mathbb{E}_{\boldsymbol{\omega}\sim\mathrm{Unif}(\{0,1\}^m)}\bigl[\,|g(\boldsymbol{\omega})|^2\,\bigr].
    \label{eq:estimator-identity}
\end{equation}
\end{restatable}

The integrand $|g(\boldsymbol{\omega})|^2$ is non-negative, so the average is free of the sign problem that affects the underlying complex amplitudes. The identity follows from the final Hadamard layer of $U_\mathcal{M}$, which turns the hidden sum in Eq.~\eqref{eq:marginal_single} into a Walsh-Hadamard transform over the hidden register. Parseval's identity then equates the sum of squared amplitudes with the average of $|g(\boldsymbol{\omega})|^2$, as we derive in App.~\ref{app:derivation}.

\begin{corollary}[Monte Carlo estimator]
\label{cor:estimator}
For $K$ independent samples $\boldsymbol{\omega}^{(k)}\sim\mathrm{Unif}(\{0,1\}^m)$, the estimator
\begin{equation*}
    \widehat{\mathbb{P}}_{\boldsymbol{\theta}}(0^n\mid x)
    = \frac{1}{K}\sum_{k=1}^{K}\bigl|g(\boldsymbol{\omega}^{(k)})\bigr|^2
\end{equation*}
is unbiased, with variance $\mathrm{Var}_{\boldsymbol{\omega}}[|g|^2]/K$ and cost $O(K\,n\,d_{\max})$. Gradients follow by reverse-mode automatic differentiation through $g$.
\end{corollary}

The number of samples needed for a target relative accuracy depends on the spread of $|g(\boldsymbol{\omega})|^2$ and on the magnitude of $\mathbb{P}_{\boldsymbol{\theta}}(0^n\mid x)$, not on $m$.

The construction of $g$ over the visible register follows the classical algorithm of Bravyi et~al.~\cite{bravyi2021} for bipartite forrelation, which builds on the Monte Carlo overlap method of Van den Nest~\cite{vandennest2010}. That algorithm estimates a single forrelation overlap by importance sampling, with variance bounded by a constant. The training signal we require is instead a marginal probability, a sum of squared amplitudes over the hidden register. Estimating this sum through the same amplitude ratio averages signed contributions and reintroduces cancellations, which Eq.~\eqref{eq:estimator-identity} avoids by averaging the non-negative $|g(\boldsymbol{\omega})|^2$ under uniform sampling.

\begin{figure*}[!t]
\centering
\begin{minipage}[b]{0.48\textwidth}
\includegraphics[width=\linewidth]{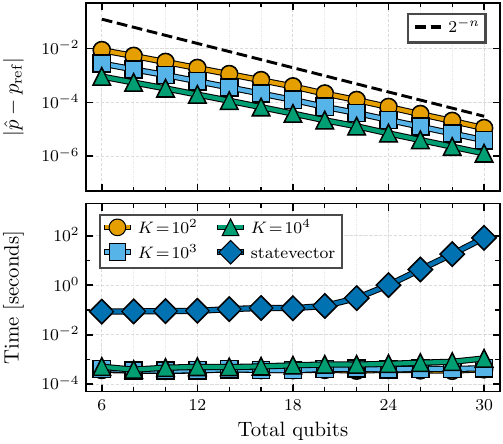}\\[2pt]
{\small (a) Estimator accuracy}
\end{minipage}
\hfill
\begin{minipage}[b]{0.48\textwidth}
\includegraphics[width=\linewidth]{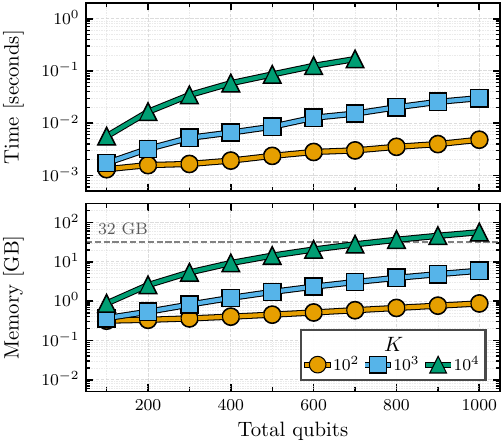}\\[2pt]
{\small (b) Scaling performance}
\end{minipage}
\caption{
\textbf{Accuracy and scaling of the Monte Carlo estimator for estimating the cost function.}
\textbf{(a)} The top panel shows the mean absolute error $|\hat{p}-p_{\mathrm{ref}}|$ for estimating $\hat{p} = \widehat{\mathbb{P}}_{\boldsymbol{\theta}}(0^n\mid x)$, with $p_{\mathrm{ref}}$ obtained from exact statevector simulation.
The dashed line marks the natural probability scale $2^{-n}$.
The bottom panel shows the corresponding wall-clock evaluation time, compared with exact statevector evaluation.
For each total system size and each Monte Carlo (MC) sample count $K$, results are averaged over $100$ independently initialized training circuits, with $100$ independent MC repetitions per circuit.
\textbf{(b)} The top panel shows the wall-clock time for a single evaluation of the visible marginal probability beyond the statevector-accessible regime.
The bottom panel shows a conservative peak-memory estimate based on the leading $O(Knm)$ storage cost of the estimator.
The dashed horizontal line indicates a 32 GB memory budget.
In both panels we set $n=m$, so that the total number of qubits $n+m$ is $2n$.
}
\label{fig:estimator-numerics}
\end{figure*}


\section{Numerical results}
\label{sec:numerics}

We first validate the MC estimator that provides the computational backbone of our training procedure.
The loss function requires visible marginal probabilities after tracing out the hidden register, but exact statevector evaluation of these quantities is limited to systems with only a few tens of qubits.
The estimator must therefore be accurate in the regime where exact comparison is possible and scalable beyond the regime accessible to exact simulation.
Accuracy of the estimator and its scalability are evaluated in the following, before discussing the model training.
The remaining results assess the trained generative model on univariate, bivariate and multivariate target distributions, supported by hardware sampling demonstrations using circuits trained with the classical estimator-based procedure. We provide additional details regarding all numerical results in App.~\ref{app:numerics}

\subsection{Estimator accuracy and scaling}
\label{sec:numerics-estimator}

To test estimator accuracy, we compare the estimator with exact statevector simulation on randomly initialized training circuits.
The circuit structure is the same as in the training objective and the parameters are drawn from a zero-mean normal distribution with standard deviation $\pi$. 
This produces broadly distributed diagonal phases and provides a useful stress test, since the resulting visible marginal probabilities can be anti-concentrated.
We set $n=m$, so that the total number of qubits is $n+m$, and estimate $\mathbb{P}_{\boldsymbol{\theta}}(0^n\mid x=0)$, with the hidden register marginalized.
For each system size and MC sample count $K$, we generate $100$ independently initialized circuits.
For each such circuit, we repeat the MC estimation $100$ times with independent estimator seeds.
This separates circuit-to-circuit variability from the sampling variance of the estimator.
The error bars are small relative to the plotted variation across $n$. 
We omit them from the main figure for readability and report them in App.~\ref{app:estimator-numerics}.

Fig.~\ref{fig:estimator-numerics}\textcolor{red}{a} shows that the estimator remains accurate throughout the statevector-accessible regime.
The mean absolute error $|\hat p-p_{\mathrm{ref}}|$ stays below the natural probability scale $2^{-n}$ for all tested sample counts $K$.
This is the relevant scale for likelihood evaluation near random initialization, where the probability mass is spread over many visible configurations.
The lower panel shows the corresponding runtime.
As expected, exact statevector evaluation grows rapidly with the total number of qubits, while the MC estimator remains nearly flat over the same range.
Reported timings exclude one-time JAX compilation costs and are measured after warm-up.

\begin{figure*}[!t]
\centering
\begin{minipage}[b]{0.48\textwidth}
\includegraphics[width=\linewidth]{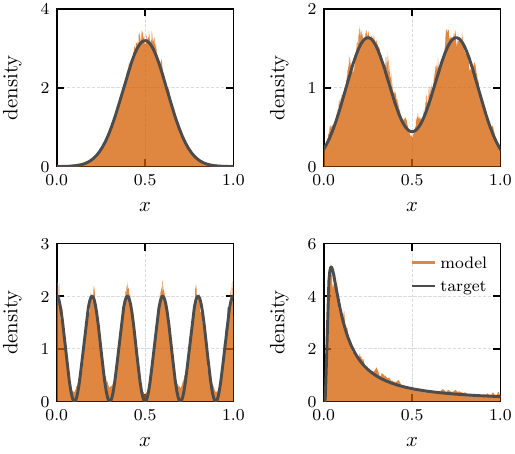}\\[2pt]
{\small (a) Simulated results}
\end{minipage}
\hfill
\begin{minipage}[b]{0.48\textwidth}
\includegraphics[width=\linewidth]{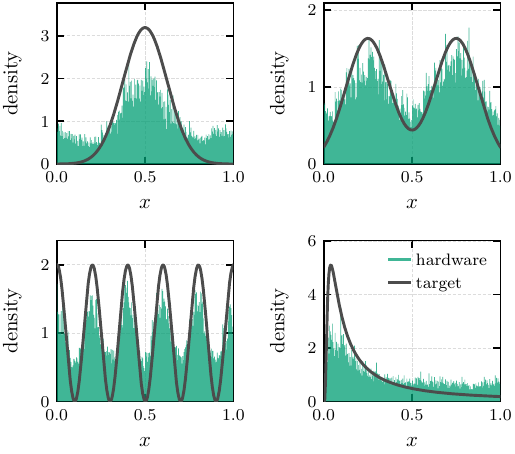}\\[2pt]
{\small (b) Hardware results}
\end{minipage}
\caption{\textbf{Learning and sampling from univariate benchmark distributions.}
\textbf{(a)} Learned model densities for four target distributions on the \texttt{float16} grid over $x\in[0,1]$, containing $2^{16}$ discrete points.
The orange filled curves show the exact trained model densities on the discrete support, and the dark curves show the target densities.
From top left to bottom right, the targets are a unimodal Gaussian, a two-component Gaussian mixture, an oscillatory cosine density, and a truncated L\'evy distribution.
All models use $m=4$ hidden qubits and are trained with the same initialization and hyperparameters.
Both the target discretization and the generated model distribution use \texttt{float16} precision.
Probabilities are rescaled by the bin width and displayed as densities on $[0,1]$.
\textbf{(b)} Sampling classically trained univariate models on quantum hardware.
Samples obtained from trained DQGM circuits executed on \texttt{ibm\_aachen}.
The dark curves show the target densities and the filled curves show hardware sample histograms rescaled as densities.
Each panel uses $10^5$ shots, $8$-bit visible precision, and $m=4$ hidden qubits.
The circuits are sampled without measurement-error mitigation, zero-noise extrapolation, dynamical decoupling, twirling, or postselection.
Hardware execution uses a gate-reduced sampling circuit, and App.~\ref{app:hardware} reports noiseless results for both the exact trained circuits and the gate-reduced circuits to separate circuit-approximation effects from hardware noise.}
\label{fig:univarite-numerics}
\end{figure*}

We next use the estimator outside the statevector-accessible regime.
Since exact reference values are no longer available, this experiment measures resource requirements rather than accuracy.
Fig.~\ref{fig:estimator-numerics}\textcolor{red}{b} shows the wall-clock time for a single evaluation of the same visible marginal probability up to $1000$ total qubits, again with $n=m$.
On a single GPU 32 GB, the runtime remains below one second for all tested values of $K$, including $K=10^4$.
The lower panel reports a conservative memory estimate based on the leading $O(Knm)$ storage cost of the estimator.
The measured memory use is lower in our implementation, but it includes backend-dependent overheads and is therefore reported separately in App.~\ref{app:estimator-numerics}.
Under this conservative estimate, $K=100$ and $K=1000$ remain below the 32 GB limit throughout the tested range, while $K=10,000$ reaches the limit near the largest system sizes.
Together, the two benchmarks show that the estimator agrees with exact simulation where such a comparison is possible and remains practical at system sizes far beyond statevector simulation.

\subsection{Learning univariate benchmark distributions}

Having established the accuracy and scaling of the estimator, we next use it for learning.
We begin with univariate benchmark distributions, where the learned density can be compared to the target on a high-resolution grid.
Fig.~\ref{fig:univarite-numerics}\textcolor{red}{a} shows four targets defined on a $2^{16}$-point discretization of $x\in[0,1]$.
The four targets are a localized Gaussian, a two-component Gaussian mixture, an oscillatory cosine density, and a truncated L\'evy distribution.
They are chosen to probe different features of a density, including localization, separated modes, oscillatory structure, and asymmetric heavy-tailed decay.

All runs use the same architecture with $m=4$ hidden qubits, the same initialization, the same hyperparameters, and the NLL objective defined in Eq.~\eqref{eq:nll}.
Each univariate training run completes in under a minute on a single GPU.
The target discretization and the generated model distribution are represented at \texttt{float16} precision.
The curves in Fig.~\ref{fig:univarite-numerics}\textcolor{red}{a} show the exact trained model distribution on the discrete support, rather than finite-sample histograms.
This setup tests whether a fixed training protocol can adapt to qualitatively different target distributions.

The trained densities reproduce the dominant structure of all four targets.
The model resolves the two modes of the Gaussian mixture, follows the oscillatory modulation of the cosine density, and captures the right-skewed tail of the heavy-tailed target.
Across the four runs, the final errors are Kullback-Leibler divergence $\mathrm{(KLD)}\in[0.0024,0.0313]$ and and total variation distance $\mathrm{(TVD)}\in[0.0206,0.0653]$, with the largest error occurring for the oscillatory target.
Per-target values and learning curves are reported in App.~\ref{app:univariate-details}.

\begin{figure*}[!t]
    \centering
    \includegraphics[width=\linewidth]{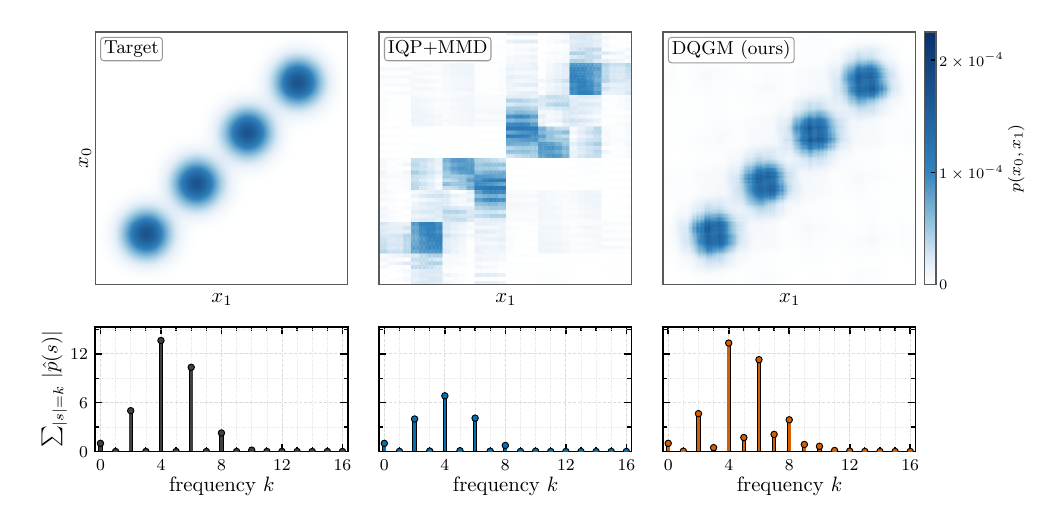}
    \caption{\textbf{Bivariate benchmark against an IQP+MMD baseline.}
    Comparison on a four-peak Gaussian target with an $8$-bit discretization for each coordinate.
    The top row shows the target distribution, the IQP+MMD baseline, and the trained DQGM.
    The IQP baseline contains all one-qubit $R_Z$ rotations and all two-qubit $R_{ZZ}$ rotations on the visible and hidden qubits ($m=4$) , giving $210$ trainable parameters.
    The DQGM uses $m=4$ hidden qubits and has $180$ trainable parameters.
    The IQP model is trained with a multi-kernel MMD objective (computed exactly) using bandwidths $\sigma\in\{0.01,0.1,1.0\}$.
    The bottom row shows the total absolute Walsh-Hadamard spectral weight grouped by Hamming weight, $\sum_{|s|=k}|\hat p(s)|$. For both models we plot the best performing seed.
    Quantitative metrics averaged over multiple seeds and training details are reported in App.~\ref{app:2d-iqp}}
    \label{fig:iqp-benchmark}
\end{figure*}

\subsection{Sampling on quantum hardware}

We finally test whether classically trained circuits can be sampled on present quantum hardware.
The circuits are trained on classical hardware as before, and only the final sampling step is executed on the quantum device.
With current devices, the main practical limitation is the quality of sampling, which is affected by  circuit depth, decoherence, gate errors, and readout errors.
We use the same family of univariate targets as in Fig.~\ref{fig:univarite-numerics}\textcolor{red}{a}, but reduce the visible precision to $8$ bits to lower the total gate count of the hardware circuits.
Each model again uses $m=4$ hidden qubits.

For hardware execution, we compile a gate-reduced version of the trained sampling circuit.
This uses an approximate inverse QFT and removes small-angle rotations to reduce the circuit depth.
We sample each model on \texttt{ibm\_aachen} with $10^5$ shots per target, with each run completing in $30$ seconds.
We report raw sample histograms, without any error mitigation or correction procedure, apart from the backend's native support for fractional gates, which is useful for the $\mathrm{RZ}$ and $\mathrm{R_{ZZ}}$ rotations appearing in the model. 
To separate algorithmic approximation from hardware noise, we report noiseless sampling results for both the exact trained circuits and the gate-reduced circuits used for hardware execution in App.~\ref{app:hardware}.

Fig.~\ref{fig:univarite-numerics}\textcolor{red}{b} shows that the hardware samples retain the dominant learned structures despite device noise.
The Gaussian targets remain localized around the expected regions, the oscillatory target preserves part of its modulation, and the L\'evy target retains its asymmetric support.
As expected, the agreement is not perfect without error mitigation.
These results indicate that, for the system sizes considered here, the main remaining bottleneck is not classical training, but accurate sampling from the trained circuit on tested quantum hardware.

\subsection{Bivariate benchmark distributions}

\begin{figure}[!t]
    \centering
    \includegraphics[width=\linewidth]{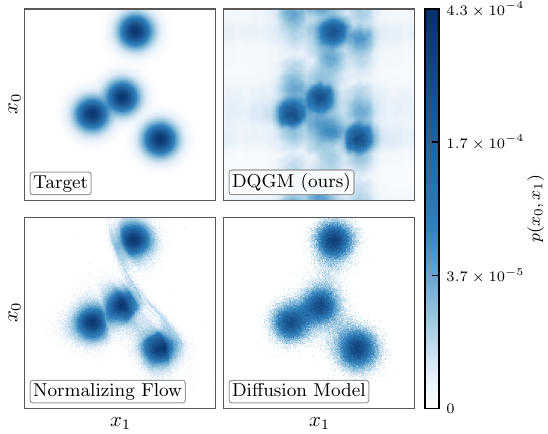}
    \caption{\textbf{Bivariate benchmark against classical generative baselines.}
    Comparison on a target distribution consisting of separated Gaussian peaks arranged along a spiral.
    All models use the same discrete data access model with an $8$-bit discretization for each coordinate.
    The panels show the target distribution, the trained DQGM, a normalizing-flow baseline, and a diffusion-model baseline. Distributions are plotted with a power-law color scaling. Quantitative metrics and baseline training details are reported in App.~\ref{app:2d-classical}}
    \label{fig:classical-benchmark}
\end{figure}

We next consider bivariate target distributions.
These experiments have two distinct roles.
The first benchmark, shown in Fig.~\ref{fig:iqp-benchmark}, compares against an IQP model trained with an MMD objective.
This is the closest point of comparison for our setting to the train-classical, deploy-quantum generative approach.
The second benchmark, shown in Fig.~\ref{fig:classical-benchmark}, compares against normalizing-flow and diffusion baselines on a deliberately difficult spiral target.
This comparison is not intended as an advantage claim over optimized classical models, but as a realistic and competitive test for the proposed model with respect to current generative baselines, under the same discrete data access model.

For the first benchmark, we train both models on the same four-peak Gaussian target with an $8$-bit discretization of each coordinate.
The DQGM uses $m=4$ hidden qubits and has $180$ trainable parameters.
For a fair comparison, we match the IQP+MMD baseline to this setting and sweep a range of configurations and perform a dedicated hyperparameter optimization that test the limits of the approach.
We vary the number of hidden qubits, the maximum gate order (up to four-qubit parametrized gates), and the way the MMD is estimated, either by MC, which is classically scalable, or exactly, which is accurate but prohibitive to assess at scale.
We report all of these configurations in App.~\ref{app:2d-iqp}.

Fig.~\ref{fig:iqp-benchmark} reports the configuration closest to the DQGM setting.
This IQP+MMD model uses $4$ hidden qubits and contains all one-qubit $R_\mathrm{Z}$ rotations and all two-qubit $R_{\mathrm{ZZ}}$ rotations, giving $210$ trainable parameters.
The parameter counts of the two models are therefore comparable, and both circuits use gates of at most two-qubits.
To avoid tuning the baseline to a single scale, we train the IQP model with a multi-kernel MMD objective using bandwidths $\sigma\in\{0.01,0.1,1.0\}$.
Although exact MMD estimation does not scale, we report the exactly trained result here to remove estimator accuracy from the comparison. 

Fig.~\ref{fig:iqp-benchmark} shows that the IQP+MMD baseline captures part of the coarse structure but produces block-like artifacts and does not reproduce the four localized peaks as accurately as the DQGM. We attribute this discrepancy to two limitations of the IQP+MMD approach.

First, IQP circuits are limited in expressivity~\cite{kurkin2025} and therefore require more gates and hidden qubits than the DQGM to represent the same distribution.
In App.~\ref{app:2d-iqp}, we show that the IQP+MMD model matches the DQGM only when it uses up to four-qubit gates and four hidden qubits and is trained exactly.
This model has $6195$ trainable parameters and fails to train with the MC estimator.

Second, the Walsh-Hadamard (WH) spectrum accounts for the remaining gap, for which we provide a primer in App.~\ref{app:mmd}.
The IQP+MMD approach assumes that the WH spectrum of the target distribution concentrates at low Hamming weight.
This assumption need not hold for multivariate, correlated distributions, as the four-peak target illustrates.
The IQP+MMD model therefore cannot match the WH spectrum at high orders.
As the lower panel of Fig.~\ref{fig:iqp-benchmark} shows, its spectrum matches the target up to Hamming weight $k=3$ and then deviates for larger $k$.
The DQGM does not rely on this assumption and matches all dominant frequencies, as evidenced by both the WH spectrum and the joint probability distributions.

The second bivariate benchmark, shown in Fig.~\ref{fig:classical-benchmark}, uses a target consisting of four separated Gaussian modes.
This target combines isolated modes with a non-Cartesian geometry, which makes it prone to oversmoothing or mode loss in generative models.
We compare against normalizing-flow and diffusion baselines trained under the same discrete data access model.

For each classical baseline we performed a dedicated hyperparameter optimization rather than adopting an off-the-shelf or toy configuration.
We optimized each baseline at three sizes and report the largest in Fig.~\ref{fig:classical-benchmark}, which for both baselines exceeds $10^4$ parameters.
The smaller models, together with quantitative metrics for all configurations, are reported in App.~\ref{app:2d-classical}.

Even with this tuning, none of the models cleanly represents the target.
The normalizing-flow baseline connects neighboring modes with spurious bridging density, and the diffusion baseline broadens and merges the adjacent modes.
The DQGM, however, better preserves the multimodal nature of the target distribution, maintaining a clear separation between the four modes and recovering its overall structure.
While residual errors and block-like artifacts are still visible, they do not compromise the model's ability to capture the essential features of the distribution.

\begin{figure}[!t]
    \centering
    \includegraphics[width=\linewidth]{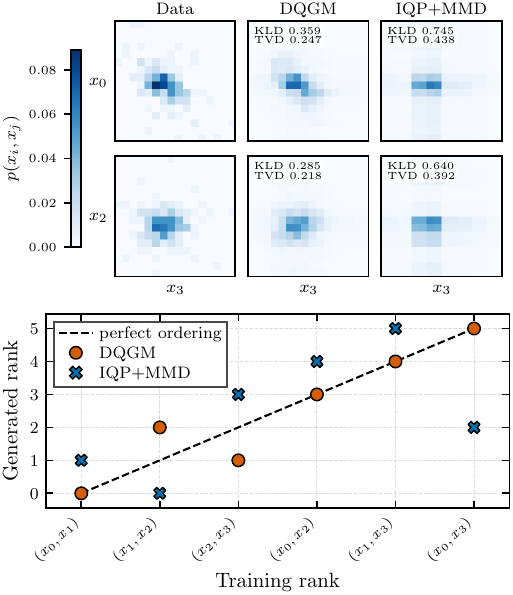}
    \caption{\textbf{Pairwise correlation structure on the four-feature finance dataset.}
    \textbf{(top)} Two-dimensional marginals $p(x_i, x_j)$ for the feature pairs $(x_0, x_3)$ and $(x_2, x_3)$, comparing the training data (left) with samples from the trained DQGM (middle) and the IQP+MMD baseline (right).
    Each model panel reports its Kullback-Leibler divergence (KLD) and total variation distance (TVD) to the corresponding two-feature projection, and all panels share the colour scale shown on the left. All pairs are reported in App.~\ref{app:supp-numerics}.
    \textbf{(bottom)} Rank preservation of the pairwise Pearson correlations.
    Feature pairs are ordered on the horizontal axis by their rank in the training data, with rank $0$ the strongest correlation, and the vertical axis gives the rank each model assigns to the same pair.
    Points lie on the dashed diagonal when a model reproduces the empirical ordering.
    The DQGM recovers the training ordering closely, while the IQP+MMD baseline departs from it for all of the pairs.}
    \label{fig:exp5b-maintext}
\end{figure}

\subsection{Multivariate benchmark distributions}

We next benchmark the model on multivariate data with correlations across features.
We use a four-feature finance dataset, encoded with $4$-bit precision per feature, so the joint distribution lives on a $16$-bit grid.
The dataset contains only $236$ samples, far fewer than the number of bins, so the empirical histogram is sparse.
In this low-data regime, the KLD and TVD summarize the fit to the sparse histogram but do not isolate the cross-feature correlations that carry the multivariate structure, so we compare the models through their pairwise correlations.
As shown in Fig.~\ref{fig:exp5b-maintext}, the DQGM recovers the ordering of all six pairwise correlations closely and attains lower KLD and TVD than the baseline, while the IQP+MMD model departs from the correlation ordering for several pairs.
The construction of the dataset and the full set of metrics are given in App.~\ref{app:supp-numerics}, which also reports a five-feature particle-physics dataset with consistent results.

Beyond the learning results we present here, we also show that the model can upsample, training at one visible precision and sampling at a higher one in App.~\ref{app:supp-numerics}. Training at $8$-bit precision and sampling at $10$ and $12$ bits, we reproduce the target with little loss of accuracy.
This allows training at lower precision and reduces the cost of classical training.

Across univariate, bivariate, and multivariate targets, our numerical study shows that likelihood-based training reconstructs the target distribution accurately.
The trained model is more faithful than the IQP+MMD baseline and remains competitive with tuned classical generative models.
The hardware runs confirm that these classically trained circuits retain their learned structure when sampled on present noisy quantum hardware, even without error mitigation.
Because we apply no mitigation or correction, the reported histograms represent a lower bound on the achievable sampling quality, and the residual gap to the target is set by hardware fidelity rather than by the training procedure.


\section{Discussion}
\label{sec:discussion}

In this work, we introduced a quantum Fourier generative model, adding a new piece in the \emph{train-on-classical deploy-on-quantum} paradigm. Trained with the negative log-likelihood and the MC estimator for visible marginal probability, we removed the exponential cost of marginalizing the hidden register, and performed training at large scale, far beyond exact statevector simulations. Our likelihood-based training reproduces the target distributions more faithfully than MMD loss. The MMD kernel fixes a low-frequency weighting of the Walsh--Hadamard spectrum, whereas the negative log-likelihood does not impose such a weighting, so the target alone sets which spectral components the loss resolves.

We note that our sampling circuit construction is closely related to IQP circuits.
Setting the per-feature precision to $d=1$ reduces the Fourier feature map and the inverse QFT to single Hadamard layers, and dropping the second diagonal layer recovers an IQP circuit.
An IQP-based QCBM is therefore a special case of our construction, and the $d$-bit, two-diagonal-layer model forms a strictly larger class.
Restricting the interactions to visible--hidden pairs, which is what makes the training signal classically tractable, need not limit expressivity, as bipartite latent-variable models such as restricted Boltzmann machines are universal approximators~\cite{Hinton2012}.
In our benchmarks, an IQP-based QCBM matches the same four-peak target only with more than thirty times as many parameters and non-scalable exact MMD training, and its low-order spectral bias otherwise leaves the high-order correlations unresolved (App.~\ref{app:2d-iqp}).

Unlike a QCBM, which uses a single circuit for both training and sampling, our construction uses two different circuits.
The training circuit yields a marginal probability that the MC estimator computes at cost $O(Knm)$, polynomial in the number of qubits, while the sampling circuit is motivated by sampling hardness, as in other train-on-classical deploy-on-quantum frameworks.
This separation makes training classically tractable but leaves open whether sampling is also classically hard.
Such hardness is established only for random instances of quantum circuits~\cite{RevModPhys.95.035001}, and likelihood training instead produces data-dependent parameters. 
Proving hardness for these parameters would be a strictly stronger statement than the anti-concentration arguments available for random ensembles, since they are fixed by the data rather than drawn at random.
The classical cost of sampling a trained circuit depends on both the target distribution and the model, not on the circuit family alone, so the same circuit can be easy to sample for one target and costly for another.
A sharper and more practical version of the question concerns wall-clock speed rather than asymptotic hardness.
A classical algorithm easily matches the roughly $300\,\mu\mathrm{s}$ per sample observed on hardware at small system sizes reached so far, and sustaining this rate at the larger, classically-hard system sizes would place quantum sampling in a competitive regime that lower-depth sampling circuits and improved hardware could bring within reach.

Where variational quantum models face a theoretical limit on what can be trained, our approach faces only a practical one on what can be sampled, set by present hardware rather than by principle. 
Classical generative modelling, by contrast, advanced through extensive empirical iteration on design choices such as connectivity, dropout, and architecture, each judged by the quality of the resulting samples. 
Our model has a comparable design space, including how hidden units connect features across frequency ranges and how the interaction graph is sparsified, and this space can be explored entirely during classical training. 
Evaluating this sample quality, however, requires running the trained circuit on hardware, which we cannot sample accurately beyond roughly $30$ qubits today, so empirical tuning of these choices is out of reach until reliable sampling at scale becomes available.
Error mitigation does not close this gap, since standard approaches such as zero-noise extrapolation and probabilistic error cancellation target expectation values, and no comparably general method exists for Born-rule sampling, where the outputs are individual samples rather than a scalar estimate. Extending error mitigation from expectation values to individual samples is an interesting avenue for future work.

As a generative model, the approach is already competitive with tuned classical baselines such as normalizing flows and diffusion models under the same discrete data access, as shown in Sec.~\ref{sec:numerics}.
What remains is sampling it at scale.
Superconducting processors with more than a thousand of qubits already exist, but their noise prevents reliable sampling at that scale.
Reducing the depth of the deployed circuit helps.
The approximate inverse QFT and gate truncation we already use preserve the learned structure while reducing depth, which shallower QFT constructions~\cite{aumann2026} could push further.
In the early fault-tolerant era, where logical error rates fall well below those of present hardware, the same classically trained circuits could be sampled with enough fidelity to resolve the learned distribution and to reopen the empirical loop that current hardware closes off.

We release our implementation at Ref.~\cite{code-repo} and the data produced during this study at Ref.~\cite{data-repo}.

\section{Methods}
\label{sec:methods}

Here, we briefly describe the notions and methods we use throughout the manuscript and refer the reader to App.~\ref{app:background} for a self-contained treatment.

\subsection{Differentiable quantum generative model}
\label{sec:dqgm}

The differentiable quantum generative model (DQGM)~\cite{kyriienko2024} is a quantum generative framework that represents probability distributions via parameterized quantum states. Training is performed by maximizing the overlap between model states and data-dependent quantum embeddings, while sampling is achieved by applying an inverse transform that maps latent states to a discrete output domain. The model therefore separates representation learning in a latent Hilbert space from the choice of sampling transform.

The DQGM operates in two modes. Training uses a data-dependent quantum circuit whose output probability at a fixed reference state $\ket{0^n}$ provides the learning signal. Sampling applies an inverse transform to the trained model state to produce discrete samples.

Let $p(x)$ be a probability distribution over a real-valued scalar $x$, and let $\mathcal{X}_\mathrm{train}$ be a dataset of $M$ i.i.d.\ samples.
Each sample is discretized onto an $n$-bit grid, yielding $x \in \{0, \dots, 2^n - 1\}$.
The Fourier feature map encodes $x$ into an $n$-qubit state via the unitary
\begin{equation}
    U_{\phi}(x) = \prod_{j=1}^n \left( \mathrm{RZ}_j \!\left( \frac{2\pi x}{2^{\,j}}\right) H_j \right),
\label{eq:encoding}
\end{equation}
where $\mathrm{RZ}_j(\varphi) = \exp(-i\varphi Z_j / 2)$ is a $z$-rotation on qubit $j$ and $H_j$ is the Hadamard gate acting on the same qubit.
Applied to the initial state $\ket{0^n}$, the encoding produces the latent state $\ket{\tilde{x}} = U_\phi(x)\ket{0^n}$.
This parametrization corresponds to the binary phase encoding underlying the quantum Fourier transform.

In this work, we choose the tranform $\phi$ to be the Fourier transform such that, for data represented on the $n$-bit discretization grid, the inverse quantum Fourier transform (QFT)~\cite{nielsen_chuang_2010} maps each encoded state $\ket{\tilde{x}}$ to the corresponding computational basis state $\ket{x}$. The sampling circuit is therefore
\begin{equation}
    \ket{\psi_{\boldsymbol{\theta}}} = \mathrm{QFT}^\dagger\, U_\mathcal{M}(\boldsymbol{\theta})\, \ket{0^n},
\label{eq:sampling}
\end{equation}
where $U_\mathcal{M}(\boldsymbol{\theta})$ denotes the variational model unitary and a computational basis measurement of $\ket{\psi_{\boldsymbol{\theta}}}$ yields samples distributed as $p_{\boldsymbol{\theta}}(x) = \bigl\lvert \langle x | \psi_{\boldsymbol{\theta}} \rangle \bigr\rvert^2$. This defines a model distribution over the discretized domain.

More generally, the transform defining $\phi$ is a design choice. Alternative transforms, such as the Chebyshev transform \cite{williams2023} or Hartley transform \cite{Wu2025}, can be used to construct a DQGM and may introduce inductive biases that are advantageous for learning specific classes of distributions, as discussed in Ref.~\cite{martinezdelejarza2025}.

\subsection{Classical and quantum separations}
\label{sec:iqp}

Instantaneous quantum polynomial (IQP) circuits~\cite{shepherd2009} are a restricted class of commuting quantum circuits and one of the most studied models for sampling-based quantum advantage~\cite{bremner2010, bremner2016, RevModPhys.95.035001}. 

\begin{definition}[IQP circuits~\cite{shepherd2009}]
\label{def:iqp}
An IQP circuit on $n$ qubits prepares
\begin{equation}
    \ket{\psi_\mathrm{IQP}} = H^{\otimes n}\Bigl(\prod_{s\in\mathcal{S}} e^{-i\theta_s Z(s)}\Bigr) H^{\otimes n}\ket{0^n},
\label{eq:iqp}
\end{equation}
where $\mathcal{S}\subseteq\{0,1\}^n$, $Z(s)=\bigotimes_{j=1}^n Z^{s_j}$ with $Z^0=I$ and $Z^1=Z$. The unitary consists of commuting diagonal operators in the computational basis, conjugated by Hadamard layers.
\end{definition}

The forrelation problem concerns studying structural similarity of IQP circuits, designed to showcase query-complexity separation of quantum and classical approaches~\cite{Aaronson2010, Aaronson2018, Umeano2026forrelation}. Both rely on alternating Hadamard layers and diagonal operators, and our construction builds on this structure.

\begin{definition}[Forrelation circuit]
\label{def:forrelation}
The forrelation circuit is
\begin{equation*}
    U_\mathrm{forr} = H^{\otimes n}\, D_g\, H^{\otimes n}\, D_f\, H^{\otimes n},
\end{equation*}
where $D_f$ and $D_g$ are diagonal unitaries with $D_f\ket{x} = f(x)\ket{x}$ and $D_g\ket{x} = g(x)\ket{x}$. Its amplitude reproduces the forrelation,
\begin{equation*}
    \bra{0^n} U_\mathrm{forr} \ket{0^n} = \Phi(f,g). 
\end{equation*}
\end{definition}
Bravyi et~al.~\cite{bravyi2021} showed that such amplitudes can be computed classically in time $O(n\,2^{w})$ up to polynomial factors when the interaction graphs of $f$ and $g$ have treewidth at most $w$. 
Motivated by this efficient classical evaluation under bounded treewidth, we designed our loss function through the same circuit structure, which we evaluate with a classical algorithm tailored to our design.

\begin{acknowledgments}
The authors acknowledge helpful suggestion from Sergey Bravyi, as well as fruitful discussions with Joseph Bowles and Zoltan Zimboras. C.T. and M.G. are supported by CERN through the CERN Quantum Technology Initiative. O.K. acknowledges the funding from UK EPSRC awards under the Agreements No. EP/Y005090/1 and No. EP/Z53318X/1 (QCi3 Hub). This work is part of the Quantum Computing for High-Energy Physics (QC4HEP) working group.
\end{acknowledgments}

\bibliographystyle{apsrev4-2}
\bibliography{main}

\onecolumngrid
\appendix

\newpage
\clearpage
\tableofcontents
\clearpage

\newpage
\section{Background}
\label{app:background}


\subsection{Differentiable quantum generative model}
\label{app:dqgm}

The differentiable quantum generative model (DQGM)~\cite{kyriienko2024} is a quantum generative framework that represents probability distributions via parameterized quantum states. Training is performed by maximizing the overlap between model states and data-dependent quantum embeddings, while sampling is achieved by applying an inverse transform that maps latent states to a discrete output domain. The model therefore separates representation learning in a latent Hilbert space from the choice of sampling transform.

The DQGM operates in two modes. Training uses a data-dependent quantum circuit whose output probability at a fixed reference state $\ket{0^n}$ provides the learning signal (also referred as explicit modelling \cite{Kasture2023}). Sampling applies an inverse transform to the trained model state to produce discrete samples. These circuits can be seen in Fig.~\ref{fig:circuits}.

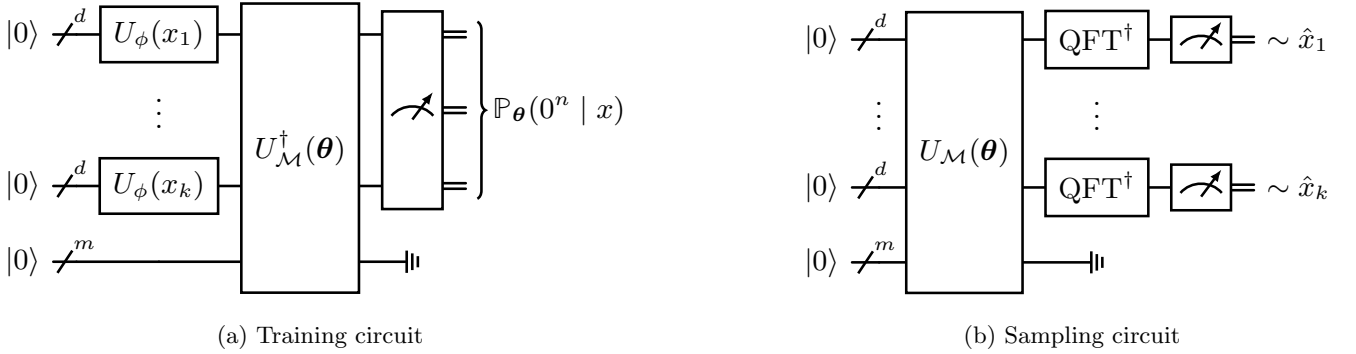
\begin{figure}[!h]
\centering
\begin{minipage}[b]{0.48\textwidth}
\centering
\resizebox{\linewidth}{!}{%
\begin{quantikz}[column sep=0.25cm, row sep=0.15cm]
\lstick{$\ket{0}$} & \qwbundle{d} & \gate{U_\phi(x_1)} & \gate[4]{U_{\mathcal{M}}^\dagger(\boldsymbol{\theta})} & \meter[3]{} & \setwiretype{c} \rstick[wires=3]{$\mathbb{P}_{\boldsymbol{\theta}}(0^n \mid x)$}\\
\setwiretype{n} & & \raisebox{1.0ex}{$\vdots$} & \ghost{U_{\mathcal{M}}^\dagger} & & \setwiretype{c}\\
\lstick{$\ket{0}$} & \qwbundle{d} & \gate{U_\phi(x_k)} & \ghost{U_{\mathcal{M}}^\dagger} & & \setwiretype{c}\\
\lstick{$\ket{0}$} & \qwbundle{m} & & \ghost{U_{\mathcal{M}}^\dagger} & \ground{} \\
\end{quantikz}
}\\[2pt]
{\small (a) Training circuit}
\end{minipage}
\hfill
\begin{minipage}[b]{0.41\textwidth}
\centering
\resizebox{\linewidth}{!}{%
\begin{quantikz}[column sep=0.25cm, row sep=0.15cm]
\lstick{$\ket{0}$} & \qwbundle{d} & \gate[4]{U_{\mathcal{M}}(\boldsymbol{\theta})} & \gate{\mathrm{QFT}^\dagger} & \meter{} & \setwiretype{c} \rstick[]{$\sim \hat{x}_1$}\\
\setwiretype{n} & \raisebox{1.0ex}{$\vdots$} & \ghost{U_{\mathcal{M}}^\dagger} & \raisebox{1.0ex}{$\vdots$} & & \\
\lstick{$\ket{0}$} & \qwbundle{d} & \ghost{U_{\mathcal{M}}^\dagger} & \gate{\mathrm{QFT}^\dagger} & \meter{} & \setwiretype{c} \rstick[]{$\sim \hat{x}_k$}\\
\lstick{$\ket{0}$} & \qwbundle{m} & \ghost{U_{\mathcal{M}}^\dagger} & \ground{} \\
\end{quantikz}
}\\[2pt]
{\small (b) Sampling circuit}
\end{minipage}

\caption{\textbf{Training and sampling circuits.}
Both circuits act on $k$ features, each encoded with $d$-bit precision, for a total of $n=kd$ visible qubits, together with $m$ hidden qubits.
In the training circuit (a), all visible qubits are measured and the hidden register is marginalized to obtain $\mathbb{P}_{\boldsymbol{\theta}}(0^n \mid x)$.
In the sampling circuit (b), an inverse QFT is applied to each feature block and the hidden register is discarded.
Each computational-basis measurement then yields one sample from the learned distribution.}
\label{fig:circuits}
\end{figure}

Let $p(x)$ be a probability distribution over a real-valued scalar $x$, and let $\mathcal{X}_\mathrm{train}$ be a dataset of $M$ i.i.d.\ samples.
Each sample is discretized onto an $n$-bit grid, yielding $x \in \{0, \dots, 2^n - 1\}$.
The Fourier feature map encodes $x$ into an $n$ qubit state via the unitary
\begin{equation}
    U_{\phi}(x) = \prod_{j=1}^n \left( \mathrm{RZ}_j \!\left( \frac{2\pi x}{2^{\,j}}\right) H_j \right),
\label{eq:app:encoding}
\end{equation}
where $\mathrm{RZ}_j(\varphi) = \exp(-i\varphi Z_j / 2)$ is a $z$-rotation on qubit $j$ and $H_j$ is the Hadamard gate acting on the same qubit.
Applied to the initial state $\ket{0^n}$, the encoding produces the latent state $\ket{\tilde{x}} = U_\phi(x)\ket{0^n}$.
This parametrization corresponds to the binary phase encoding underlying the quantum Fourier transform.

Let $U_\mathcal{M}(\boldsymbol{\theta})$ denote the variational model unitary.
For a given data point $x$, we define the sample-wise objective as the squared overlap
\begin{equation}
    \mathcal{J}_{\boldsymbol{\theta}}(x) 
    = \bigl|\bra{0^n} U_\mathcal{M}^\dagger(\boldsymbol{\theta})\, U_\phi(x) \ket{0^n}\bigr|^2,
\end{equation}
which corresponds to the Born probability of measuring the all-zero outcome after applying the inverse model to the encoded data,
\begin{equation}
    \mathcal{J}_{\boldsymbol{\theta}}(x) = \mathbb{P}_{\boldsymbol{\theta}}(0^n \mid x).
\end{equation}
Equivalently, this can be interpreted as the overlap between the model state $\ket{\psi_{\boldsymbol{\theta}}} = U_\mathcal{M}(\boldsymbol{\theta})\ket{0^n}$ and the encoded data state $\ket{\tilde{x}}$.

The population objective is given by the expectation over the data distribution
\begin{equation}
    \mathcal{J}(\boldsymbol{\theta}) 
    = \mathbb{E}_{x \sim p(x)} \big[ \mathcal{J}_{\boldsymbol{\theta}}(x) \big]. 
\end{equation}
In practice, this expectation is approximated by the empirical average over the training dataset $\mathcal{X}_{\mathrm{train}}$,
\begin{equation}
    \bar{\mathcal{J}}(\boldsymbol{\theta}) 
    = \frac{1}{M} \sum_{x \in \mathcal{X}_{\mathrm{train}}} \mathcal{J}_{\boldsymbol{\theta}}(x).
\end{equation}

This objective can be interpreted as a likelihood surrogate in the latent space induced by the feature map $\phi$. More generally, the choice of loss function is a design decision and may be adapted to the task at hand. In all cases, the underlying goal is to maximize the overlap between the model state and the encoded data states.

Alternatively, one may consider the negative log-likelihood
\begin{equation}
    \mathcal{L}(\boldsymbol{\theta}) 
    = - \frac{1}{M} \sum_{x \in \mathcal{X}_{\mathrm{train}}} 
    \log \mathcal{J}_{\boldsymbol{\theta}}(x),
\end{equation}
which often yields more informative gradients and improved numerical stability.

In the present construction, the Fourier feature map is chosen such that, for data represented on the $n$-bit discretization grid, the inverse quantum Fourier transform (QFT)~\cite{nielsen_chuang_2010} maps each encoded state $\ket{\tilde{x}}$ to the corresponding computational basis state $\ket{x}$. The sampling circuit is therefore
\begin{equation}
    \ket{\psi_{\boldsymbol{\theta}}} = \mathrm{QFT}^\dagger\, U_\mathcal{M}(\boldsymbol{\theta})\, \ket{0^n},
\label{eq:app:sampling}
\end{equation}
and a computational basis measurement of $\ket{\psi_{\boldsymbol{\theta}}}$ yields samples distributed as $p_{\boldsymbol{\theta}}(x) = \bigl\lvert \langle x | \psi_{\boldsymbol{\theta}} \rangle \bigr\rvert^2$. This defines a model distribution over the discretized domain.

\subsection{IQP circuits and forrelation}
\label{app:iqp}

Instantaneous quantum polynomial (IQP) circuits~\cite{shepherd2009} are a restricted class of commuting quantum circuits and one of the most studied models for sampling-based quantum advantage~\cite{bremner2010, bremner2016, RevModPhys.95.035001}. The forrelation problem is a query-complexity separation whose quantum algorithm uses a structurally similar circuit~\cite{Aaronson2010, Aaronson2018}. Both rely on alternating Hadamard layers and diagonal operators, and our construction builds on this structure.

\begin{definition}[IQP circuits~\cite{shepherd2009}]
An IQP circuit on $n$ qubits prepares
\begin{equation}
    \ket{\psi_\mathrm{IQP}} = H^{\otimes n}\Bigl(\prod_{s\in\mathcal{S}} e^{-i\theta_s Z(s)}\Bigr) H^{\otimes n}\ket{0^n},
\label{eq:app:iqp}
\end{equation}
where $\mathcal{S}\subseteq\{0,1\}^n$, $Z(s)=\bigotimes_{j=1}^n Z^{s_j}$ with $Z^0=I$ and $Z^1=Z$. The unitary consists of commuting diagonal operators in the computational basis, conjugated by Hadamard layers.
\end{definition}

Sampling from IQP circuits is conjectured to be classically intractable under standard complexity-theoretic assumptions, even for approximate sampling~\cite{bremner2010, bremner2016}. Two properties support this conjecture, anti-concentration of the output distribution and the hardness of approximating output amplitudes.

\begin{definition}[Anti-concentration~\cite{RevModPhys.95.035001}]
Let $\mathcal{U}$ be an ensemble of quantum circuits on $n$ qubits. The ensemble $\mathcal{U}$ anti-concentrates if there exist constants $\alpha,\,\beta > 0$, independent of $n$, such that
\begin{equation}
    \Pr_{U \sim\, \mathcal{U}}\!\left[\, p_U(x) \ge \frac{\alpha}{2^n} \right] \ge \beta
\end{equation}
for every $x \in \{0,1\}^n$, where $p_U(x) = \bigl|\bra{x} U \ket{0^n}\bigr|^2$ and the probability is over the choice of circuit $U \sim \mathcal{U}$.
\end{definition}

Anti-concentration has been proven for circuit ensembles that form approximate unitary $2$-designs, including random IQP circuits and universal random circuits~\cite{RevModPhys.95.035001}. In these cases the second moment of the output probabilities stays close to its Haar value, which prevents the distribution from concentrating on a small set of outcomes.

While IQP hardness is phrased in terms of sampling, a structurally similar circuit underlies one of the earliest~\cite{Aaronson2010} and strongest~\cite{Aaronson2018} separations between classical and quantum query complexity, the forrelation problem. The problem takes two Boolean functions $f,g:\{0,1\}^n \to \{-1,1\}$ and their forrelation,
\begin{equation}
    \Phi(f,g) = \frac{1}{2^{3n/2}} \sum_{x,y \in \{0,1\}^n} f(x)\,(-1)^{x \cdot y}\,g(y), 
\end{equation}
which measures the correlation between $f$ and the Fourier transform of $g$. The task is to decide whether $\Phi(f,g) \ge 3/5$ or $|\Phi(f,g)| \le 1/100$, promised that one of these holds. Just $O(1)$ quantum queries solve this with bounded error, whereas any classical randomized algorithm requires $\Omega(\sqrt{2^n}/n)$ queries~\cite{Aaronson2018}.

\begin{definition}[Forrelation circuit]
The forrelation circuit is
\begin{equation}
    U_\mathrm{forr} = H^{\otimes n}\, D_g\, H^{\otimes n}\, D_f\, H^{\otimes n},
\label{eq:app:forrelation}
\end{equation}
where $D_f$ and $D_g$ are diagonal unitaries with $D_f\ket{x} = f(x)\ket{x}$ and $D_g\ket{x} = g(x)\ket{x}$. Its amplitude reproduces the forrelation,
\begin{equation}
    \bra{0^n} U_\mathrm{forr} \ket{0^n} = \Phi(f,g). 
\end{equation}
\end{definition}

As anticipated, $U_\mathrm{forr}$ has the same alternating Hadamard-and-diagonal structure as the IQP circuit of Eq.~\eqref{eq:iqp}. 

\subsection{MMD loss, Pauli-Z expectations, and the Walsh--Hadamard transform}
\label{app:mmd}

A quantum circuit Born machine (QCBM) encodes a probability distribution
$q_{\boldsymbol{\theta}}(x) = |\langle x | \psi_{\boldsymbol{\theta}}\rangle|^2$ over bitstrings $x \in \{0,1\}^n$ and is trained to reproduce a target distribution $p(x)$.
A standard loss for this task is the squared Maximum Mean Discrepancy (MMD)~\cite{rudolph2024, Recio-Armengol2025}
\begin{equation}
    \mathrm{MMD}^2(p, q_{\boldsymbol{\theta}})
    = \mathbb{E}_{x,x' \sim p}\bigl[k(x,x')\bigr]
    - 2\,\mathbb{E}_{\substack{x \sim p\\ y \sim q_{\boldsymbol{\theta}}}}\bigl[k(x,y)\bigr]
    + \mathbb{E}_{y,y' \sim q_{\boldsymbol{\theta}}}\bigl[k(y,y')\bigr],
    \label{eq:app:mmd}
\end{equation}
which vanishes if and only if $p = q_{\boldsymbol{\theta}}$ for a characteristic kernel $k$.
The kernel used for QCBMs is a Gaussian (radial basis function), often a mixture over
bandwidths~\cite{Liu2018, Coyle2020}. On bitstrings the squared Euclidean distance equals
the Hamming weight of $x \oplus y$, so the kernel is translation invariant on $\mathbb{Z}_2^n$,
\begin{equation}
    k(x,y) = \kappa(x \oplus y), \qquad
    \kappa(z) = \exp\!\left(-\tfrac{|z|}{2\sigma^2}\right),
    \label{eq:app:gaussian-kernel}
\end{equation}
where $|z|$ denotes the Hamming weight of $z$ and $\sigma$ is the bandwidth.

Functions on the Boolean cube admit a Walsh--Hadamard expansion in the characters
$\chi_s(x) = (-1)^{s \cdot x}$ indexed by $s \in \{0,1\}^n$, which form an orthogonal
basis~\cite{odonnell2014}. The Walsh--Hadamard transform of the target distribution and its
inverse are
\begin{equation}
    \hat p(s) = \sum_{x \in \{0,1\}^n} p(x)\,(-1)^{s \cdot x},
    \qquad
    p(x) = \frac{1}{2^n} \sum_{s \in \{0,1\}^n} \hat p(s)\,(-1)^{s \cdot x},
    \label{eq:app:wh}
\end{equation}
and the same transform applies to the model distribution $q_{\boldsymbol{\theta}}$ and to the kernel function $\kappa$. The values $\hat p(s)$ are the Walsh--Hadamard coefficients of $p$.

Each character is the eigenvalue function of a Pauli-Z string. The operator
$Z(s) = \bigotimes_{j=1}^n Z^{s_j}$ is diagonal in the computational basis with
$Z(s)\ket{x} = (-1)^{s \cdot x}\ket{x}$, so for any Born machine with state $\rho_{\boldsymbol{\theta}}$,
\begin{equation}
    \langle Z(s) \rangle_{q_{\boldsymbol{\theta}}}
    = \mathrm{Tr}\!\bigl[Z(s)\,\rho_{\boldsymbol{\theta}}\bigr]
    = \sum_{x \in \{0,1\}^n} q_{\boldsymbol{\theta}}(x)\,(-1)^{s \cdot x}
    = \hat q_{\boldsymbol{\theta}}(s).
    \label{eq:app:pauli-fourier}
\end{equation}
The Pauli-Z string expectation values of a QCBM are therefore the Walsh--Hadamard
coefficients of its distribution. The same identity holds for the target, $\langle Z(s)\rangle_p = \hat p(s)$.

To express the loss in this spectrum, write $\delta = p - q_{\boldsymbol{\theta}}$ and expand the squares in Eq.~\eqref{eq:app:mmd},
\begin{equation}
    \mathrm{MMD}^2(p, q_{\boldsymbol{\theta}})
    = \sum_{x,y \in \{0,1\}^n} \kappa(x \oplus y)\,\delta(x)\,\delta(y).
    \label{eq:app:mmd-bilinear}
\end{equation}
Substituting the inverse transform of $\kappa$ from Eq.~\eqref{eq:app:wh} and collecting the character sum into $\hat p(s) - \hat q_{\boldsymbol{\theta}}(s)$ via Eq.~\eqref{eq:app:pauli-fourier} diagonalizes
the loss over Pauli modes,
\begin{equation}
    \mathrm{MMD}^2(p, q_{\boldsymbol{\theta}})
    = \frac{1}{2^n} \sum_{s \in \{0,1\}^n} \hat\kappa(s)\,
      \bigl(\hat p(s) - \hat q_{\boldsymbol{\theta}}(s)\bigr)^2.
    \label{eq:app:mmd-fourier}
\end{equation}
The Gaussian kernel of Eq.~\eqref{eq:app:gaussian-kernel} factorizes as
$\kappa(z) = \prod_{j=1}^n a^{z_j}$ with $a = e^{-1/2\sigma^2}$, so its spectrum depends only on
the weight $|s|$,
\begin{equation}
    \hat\kappa(s) = (1+a)^{\,n-|s|}\,(1-a)^{\,|s|}.
    \label{eq:app:kernel-spectrum}
\end{equation}
As $0 < a < 1$, the weights are positive and decrease in $|s|$, so the kernel acts as a low-pass filter over Pauli weight and suppresses the contribution of high-weight Pauli strings to the loss.

For the IQP circuits of Definition~\ref{def:iqp}, the Walsh--Hadamard basis is the natural one, since the commuting generators $e^{-i\theta_s Z(s)}$ are diagonal Pauli-Z strings and the Hadamard layers map between the computational and Fourier bases. Eq.~\eqref{eq:app:mmd-fourier} then expresses the loss as a kernel-weighted sum of squared differences of Walsh--Hadamard coefficients, equivalently of Pauli-Z expectation values, between target and model. The constant mode $s = 0^n$ does not contribute, as $\hat p(0^n) = \hat q_{\boldsymbol{\theta}}(0^n) = 1$ for any normalized distribution, consistent with
the MMD being insensitive to overall normalization.

\newpage
\section{Additional details of the DQGM model}
\label{app:model-details}

Here, we provide additional details for our model design and choices that we excluded from the main text for readbility.

While choosing $U_\mathcal{M}$, single-qubit terms are included only on visible qubits, where they incorporate both model parameters and data-dependent phase shifts. Terms acting exclusively on hidden qubits are omitted, as they contribute only global phases after marginalization over the hidden register and therefore do not affect the training objective.

The data-encoding RZ rotations of the feature map $U_\phi$ (see Eq.~\eqref{eq:encoding}) are diagonal in the computational basis, the same basis as the second diagonal layer $D(\boldsymbol{\theta}^{(2)})$. During training we therefore absorb them into $\boldsymbol{\theta}^{(2)}$, so encoding a data point shifts the trainable angles instead of adding gates. A feature $x_l$ encoded on local qubits $j \in \{ 1,\ldots, d\}$ shifts the corresponding angles as
\begin{equation}
    \theta_{l,\,j}^{(2)} \;\longrightarrow\; \theta_{l,\,j}^{(2)} + \frac{2\pi x_l}{2^{\,j}}, 
\end{equation}
where $d$ is the encoding precision per feature and the total number of visible qubits is $n = k d$ for $k$ features.

The intermediate layer of $U_\mathcal{M}$ is fixed to \( \mathrm{RY}(\pi/2) \) on visible qubits and identity on hidden qubits. This choice preserves the diagonal structure on the hidden register while allowing nontrivial mixing on the visible subsystem. As a consequence, for any fixed hidden configuration, the circuit induces independent contributions from each visible qubit.

Our construction is a Fourier-based model and therefore inherits artifacts of the Fourier transform.
In practice, we observe that a model encoding $d$-bit precision data on $d$ qubits fails to represent certain distributions, which we attribute to a Nyquist-related effect.
We resolve this by adding a single dummy qubit, which we refer to as a frequency buffer. For all results reported in this paper, we therefore use $d+1$ qubits to encode $d$-bit precision data.

\newpage
\section{Derivation of the marginal estimator}
\label{app:derivation}

\subsection{Problem and circuit}

We derive an estimator for the marginal, establishing Lemma~\ref{lem:factorization} and Proposition~\ref{prop:estimator}.

Training minimizes the negative log-likelihood of Eq.~\eqref{eq:nll}, which for each input $x$ requires the marginal probability $\mathbb{P}_{\boldsymbol{\theta}}(0^n\mid x)$. The training circuit first encodes the data with $U_\phi(x)$ of Eq.~\eqref{eq:encoding}, a Hadamard and an $\mathrm{RZ}$ rotation on each visible qubit, and then applies the inverse model unitary $U_\mathcal{M}^\dagger$. Composing the two, the Hadamards of $U_\phi$ join the hidden-register Hadamards of $U_\mathcal{M}^\dagger$ into a full layer, and the circuit becomes
\begin{equation}
    W(\boldsymbol{\theta},x) = H^{\otimes N}\,D_1\,\bigl(\mathrm{RY}(\pi/2)^{\otimes n}\otimes I^{\otimes m}\bigr)\,D_2\,H^{\otimes N},
    \label{eq:Wapp}
\end{equation}
two diagonal layers separated by a fixed visible rotation and conjugated by Hadamard layers on all $N$ qubits, where $D_2=D(\boldsymbol{\theta}^{(2)})^\dagger D_\phi(x)$ carries the data and $D_1=D(\boldsymbol{\theta}^{(1)})^\dagger$. Fig.~\ref{fig:training-circuit-app} shows the circuit, with $U_\phi$ marked as a dashed block. Both diagonal layers are bipartite, with single-qubit terms on visible qubits and two-qubit terms that each join one visible qubit to one hidden qubit. The training signal is the probability of reading $0^n$ on the visible register with the hidden register summed out,
\begin{equation}
    \mathbb{P}_{\boldsymbol{\theta}}(0^n\mid x) = \sum_{\mathbf{b}_h\in\{0,1\}^m}\bigl|\langle 0^n,\mathbf{b}_h|W(\boldsymbol{\theta},x)|0^N\rangle\bigr|^2.
    \label{eq:margapp}
\end{equation}

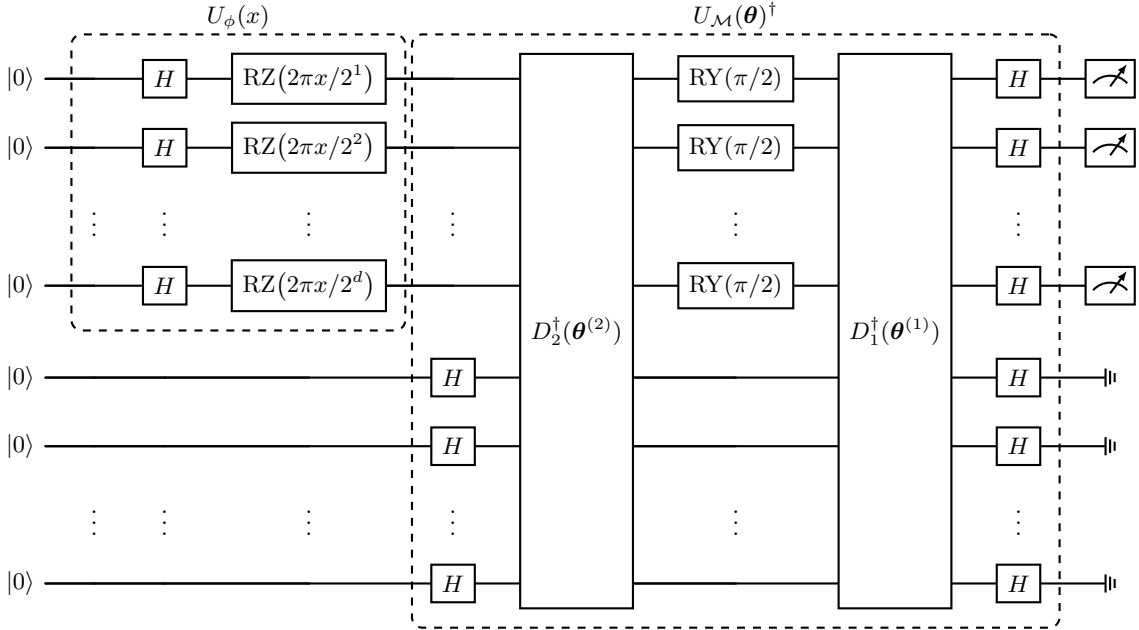
\begin{figure}[!h]
\centering
\begin{quantikz}[column sep=0.6cm, row sep=0.25cm]
\lstick{$\ket{0}$}
& \qw
  \gategroup[
    wires=4,
    steps=3,
    style={dashed, rounded corners}
  ]{$U_\phi(x)$}
& \gate{H}
& \gate{\mathrm{RZ}\!\left(2\pi x/2^1\right)}
& \qw
  \gategroup[
    wires=8,
    steps=5,
    style={dashed, rounded corners}
  ]{$U_{\mathcal{M}}(\boldsymbol{\theta})^\dagger$}
& \gate[8]{D_2^\dagger(\boldsymbol{\theta}^{(2)})}
& \gate{\mathrm{RY}(\pi/2)}
& \gate[8]{D_1^\dagger(\boldsymbol{\theta}^{(1)})}
& \gate{H}
& \meter{}
\\
\lstick{$\ket{0}$}
& \qw
& \gate{H}
& \gate{\mathrm{RZ}\!\left(2\pi x/2^2\right)}
& \qw
& \ghost{D_2^\dagger(\boldsymbol{\theta}^{(2)})}
& \gate{\mathrm{RY}(\pi/2)}
& \ghost{D_1^\dagger(\boldsymbol{\theta}^{(1)})}
& \gate{H}
& \meter{}
\\
\setwiretype{n}
& \push{\vdots}
& \push{\vdots}
& \push{\vdots}
& \push{\vdots}
& \ghost{D_2^\dagger(\boldsymbol{\theta}^{(2)})}
& \push{\vdots}
& \ghost{D_1^\dagger(\boldsymbol{\theta}^{(1)})}
& \push{\vdots}
&
\\
\lstick{$\ket{0}$}
& \qw
& \gate{H}
& \gate{\mathrm{RZ}\!\left(2\pi x/2^d\right)}
& \qw
& \ghost{D_2^\dagger(\boldsymbol{\theta}^{(2)})}
& \gate{\mathrm{RY}(\pi/2)}
& \ghost{D_1^\dagger(\boldsymbol{\theta}^{(1)})}
& \gate{H}
& \meter{}
\\[0.30cm]
\lstick{$\ket{0}$}
& \qw
& \qw
& \qw
& \gate{H}
& \ghost{D_2^\dagger(\boldsymbol{\theta}^{(2)})}
& \qw
& \ghost{D_1^\dagger(\boldsymbol{\theta}^{(1)})}
& \gate{H}
& \ground{}
\\
\lstick{$\ket{0}$}
& \qw
& \qw
& \qw
& \gate{H}
& \ghost{D_2^\dagger(\boldsymbol{\theta}^{(2)})}
& \qw
& \ghost{D_1^\dagger(\boldsymbol{\theta}^{(1)})}
& \gate{H}
& \ground{}
\\
\setwiretype{n}
& \push{\vdots}
& \push{\vdots}
& \push{\vdots}
& \push{\vdots}
& \ghost{D_2^\dagger(\boldsymbol{\theta}^{(2)})}
& \push{\vdots}
& \ghost{D_1^\dagger(\boldsymbol{\theta}^{(1)})}
& \push{\vdots}
&
\\
\lstick{$\ket{0}$}
& \qw
& \qw
& \qw
& \gate{H}
& \ghost{D_2^\dagger(\boldsymbol{\theta}^{(2)})}
& \qw
& \ghost{D_1^\dagger(\boldsymbol{\theta}^{(1)})}
& \gate{H}
& \ground{}
\end{quantikz}
\caption{\textbf{A detailed view of the training circuit.} The encoding $U_\phi(x)$ applies a Hadamard and an $\mathrm{RZ}(2\pi x/2^{\,j})$ to each visible qubit $j$. The inverse model unitary $U_\mathcal{M}^\dagger$ then applies two diagonal layers $D_1,D_2$ separated by a fixed $\mathrm{RY}(\pi/2)$ on the visible qubits, conjugated by Hadamard layers on all $N$ qubits. The training signal $\mathbb{P}_{\boldsymbol{\theta}}(0^n\mid x)$ is the probability of the all-zero visible outcome with the hidden register traced out. Here, we show the univariate (single feature) case, where $n=d$, $N=d+m$ and $x \in \{0, 1, \dots, 2^d - 1 \}$.}
    \label{fig:training-circuit-app}
\end{figure}

The sum in Eq.~\eqref{eq:margapp} runs over all $2^m$ hidden configurations, so a direct evaluation costs time exponential in $m$. Parseval's identity removes this cost, as follows. The circuit ends with a Hadamard layer on all qubits, so measuring the hidden register in the computational basis amounts to a Walsh-Hadamard transform, and summing the squared amplitude over the hidden outcomes equals a sum over the conjugate hidden variable. In that variable the visible qubits decouple, because each visible qubit couples only to hidden qubits, and the marginal becomes a uniform average of a quantity that factorizes over the visible register.

\subsection{Setup}

We use $n$ visible and $m$ hidden qubits, $N=n+m$ in total, and write a computational-basis string as $\mathbf{z}=(\mathbf{a},\boldsymbol{\omega})$ with visible part $\mathbf{a}\in\{0,1\}^n$ and hidden part $\boldsymbol{\omega}\in\{0,1\}^m$.

Following Bravyi et~al.~\cite{bravyi2021}, which builds on the Monte Carlo overlap method of Van den Nest~\cite{vandennest2010}, we split $W$ at its middle, introduce two states, and form an amplitude ratio between them. The first state collects the input Hadamards, the data-carrying diagonal, and the intermediate rotation, and the second collects the output diagonal,
\begin{equation}
    |\alpha\rangle = \bigl(\mathrm{RY}(\pi/2)^{\otimes n}\otimes I\bigr)D_2 H^{\otimes N}|0^N\rangle,
    \qquad
    |\beta\rangle = D_1^\dagger H^{\otimes N}|0^N\rangle.
\end{equation}
We record the basis amplitudes $A(\mathbf{z})=\langle\mathbf{z}|\alpha\rangle$ and $B(\mathbf{z})=\langle\mathbf{z}|\beta\rangle$, the ratio $R=\overline{A}/\overline{B}$, and its average over the visible register $g(\boldsymbol{\omega})=\mathbb{E}_{\mathbf{a}\sim\mathrm{Unif}}[R]$. The estimator is built from these objects.

Each gate in a diagonal layer contributes a phase of unit modulus. A single-qubit term of angle $\theta$ on a bit $t$ contributes $g_1(t;\theta)=\exp(-i\theta(-1)^t)$, and a two-qubit term of angle $\theta$ on bits $t,b$ contributes $g_2(t,b;\theta)=\exp(-i\theta(-1)^{t+b})$. The angles of $D_2$ include the data through $D_\phi(x)$, and the angles of $D_1$ are those of $D(\boldsymbol{\theta}^{(1)})^\dagger$.

Since $D_1$ is a diagonal unitary it only rephases each basis state, so its entries $(D_1)_{\mathbf{z}}$ have unit modulus and $B$ has constant modulus. Using $|\beta\rangle=D_2^\dagger H^{\otimes N}|0^N\rangle$ and $\langle\mathbf{z}|H^{\otimes N}|0^N\rangle=2^{-N/2}$,
\begin{equation}
    B(\mathbf{z}) = 2^{-N/2}\,\overline{(D_1)_{\mathbf{z}}},
    \qquad |B(\mathbf{z})| = 2^{-N/2}\ \text{for all }\mathbf{z}.
\end{equation}
The two identities below follow,
\begin{equation}
    (D_1)_{\mathbf{z}} = 2^{N/2}\,\overline{B(\mathbf{z})},
    \qquad
    \frac{1}{\overline{B(\mathbf{z})}} = 2^{N}\,B(\mathbf{z}).
    \label{eq:Bident}
\end{equation}
The constant modulus of $B$ also keeps $R$ finite, since $\overline{B}$ never vanishes.

\subsection{Proof of Lemma~\ref{lem:factorization}}

We begin by restating Lemma~\ref{lem:factorization} from the main text.

\factlemma*   
\begin{proof}
The circuit never couples two visible qubits, so once the hidden bits are fixed each visible qubit evolves independently, and $A$, $B$, and $g$ factorize over the visible register.

Insert $H^{\otimes N}|0^N\rangle=2^{-N/2}\sum_{\mathbf{t},\mathbf{c}}|\mathbf{t},\mathbf{c}\rangle$ into $|\alpha\rangle$. The intermediate rotation acts on visible qubits only, so $\langle\mathbf{a},\boldsymbol{\omega}|(\mathrm{RY}(\tfrac{\pi}{2})^{\otimes n}\otimes I)|\mathbf{t},\mathbf{c}\rangle=\delta_{\boldsymbol{\omega},\mathbf{c}}\prod_v[\mathrm{RY}(\tfrac{\pi}{2})]_{a_v,t_v}$. This fixes the hidden bits to $\boldsymbol{\omega}$ and is a product over visible qubits. The diagonal entry of $D_1$ at $(\mathbf{t},\boldsymbol{\omega})$ is the product of its gate phases, grouped by visible qubit since $D_1$ has no visible-to-visible term,
\begin{equation}
    (D_2)_{\mathbf{t},\boldsymbol{\omega}} = \prod_{v=1}^{n} g_1(t_v;\theta^{(2)}_v)\!\prod_{j\in\mathcal{N}(v)}\! g_2(t_v,\omega_j;\theta^{(2)}_{vj}),
\end{equation}
where $\mathcal{N}(v)$ is the set of hidden qubits coupled to $v$. The sum over $\mathbf{t}$ then factorizes, and using $[\mathrm{RY}(\tfrac{\pi}{2})]_{s,t}=2^{-1/2}(-1)^{(1-s)t}$,
\begin{equation}
    A(\mathbf{a},\boldsymbol{\omega}) = 2^{-N/2}\prod_{v=1}^{n}\mathrm{LS}_A(v,a_v;\boldsymbol{\omega}),
\end{equation}
\begin{equation}
    \mathrm{LS}_A(v,s;\boldsymbol{\omega}) = 2^{-1/2}\!\!\sum_{t\in\{0,1\}}\!\!(-1)^{(1-s)t}\,g_1(t;\theta^{(2)}_v)\!\!\prod_{j\in\mathcal{N}(v)}\!\! g_2(t,\omega_j;\theta^{(2)}_{vj}).
\end{equation}
The same steps applied to the output diagonal, which has no intermediate rotation and so no sum over $\mathbf{t}$, give a product of per-qubit factors. Because this diagonal enters the training circuit as its adjoint $D(\boldsymbol{\theta}^{(2)})^\dagger$ (Eq.~\eqref{eq:Wapp}), its single- and two-qubit phases carry the opposite sign,
\begin{equation}
    B(\mathbf{a},\boldsymbol{\omega}) = 2^{-N/2}\prod_{v=1}^{n}\mathrm{LS}_B(v,a_v;\boldsymbol{\omega}),
    \quad
    \mathrm{LS}_B(v,s;\boldsymbol{\omega}) = g_1(s;-\theta^2_v)\!\!\prod_{j\in\mathcal{N}(v)}\!\! g_2(s,\omega_j;-\theta^2_{vj}).
\end{equation}
Both $A$ and $B$ carry the same $2^{-N/2}$ prefactor, which cancels in the ratio, so $R$ becomes a product of per-qubit ratios,
\begin{equation}
    R(\mathbf{a},\boldsymbol{\omega}) = \prod_{v=1}^{n} r_v(a_v;\boldsymbol{\omega}),
    \qquad
    r_v(s;\boldsymbol{\omega}) = \overline{\mathrm{LS}_A(v,s;\boldsymbol{\omega})}\,/\,\overline{\mathrm{LS}_B(v,s;\boldsymbol{\omega})}.
\end{equation}
The visible bits $a_1,\dots,a_n$ are averaged independently and uniformly, so the average of the product is the product of the averages,
\begin{equation}
    g(\boldsymbol{\omega}) = \prod_{v=1}^{n}\frac{r_v(0;\boldsymbol{\omega})+r_v(1;\boldsymbol{\omega})}{2}.
\end{equation}
Evaluating one $r_v$ needs the two terms of $\mathrm{LS}_A$ and the single $\mathrm{LS}_B$, each a product over the $|\mathcal{N}(v)|$ hidden neighbours of $v$, at cost $O(d_{\max})$, with $d_{\max}=\max_v|\mathcal{N}(v)|$. The product over $n$ visible qubits costs $O(n\,d_{\max})$. 
\end{proof}

\subsection{Proof of Proposition~\ref{prop:estimator}}

Now let us restate Proposition~\ref{prop:estimator} from the main text.

\estimatorprop*
\begin{proof}
The amplitude of a single hidden outcome, as a function of that outcome, is a Walsh-Hadamard transform over the hidden register.

Let $a(\mathbf{b}_h)=\langle 0^n,\mathbf{b}_h|W|0^N\rangle$, so that $\mathbb{P}_{\boldsymbol{\theta}}(0^n\mid x)=\sum_{\mathbf{b}_h}|a(\mathbf{b}_h)|^2$ by Eq.~\eqref{eq:margapp}. Everything in Eq.~\eqref{eq:Wapp} up to the final Hadamards is $|\alpha\rangle$ followed by $D_1$, so $W|0^N\rangle=H^{\otimes N}D_1|\alpha\rangle$ and
\begin{equation}
    a(\mathbf{b}_h)=\langle 0^n,\mathbf{b}_h|H^{\otimes N}D_1|\alpha\rangle.
\end{equation}
The final Hadamards act on the bra. On the visible qubits the fixed outcome $0^n$ becomes a uniform superposition with no phase, while on the hidden qubits the outcome $\mathbf{b}_h$ produces the Walsh-Hadamard kernel, introducing a new hidden variable $\boldsymbol{\omega}$,
\begin{equation}
    \langle 0^n,\mathbf{b}_h|H^{\otimes N} = 2^{-N/2}\sum_{\mathbf{a},\boldsymbol{\omega}}(-1)^{\mathbf{b}_h\cdot\boldsymbol{\omega}}\,\langle\mathbf{a},\boldsymbol{\omega}|.
\end{equation}
Inserting this, using that $D_1$ is diagonal with $\langle\mathbf{a},\boldsymbol{\omega}|\alpha\rangle=A(\mathbf{a},\boldsymbol{\omega})$ and $(D_1)_{\mathbf{z}}=2^{N/2}\overline{B}$ from Eq.~\eqref{eq:Bident}, the $2^{\pm N/2}$ factors cancel and
\begin{equation}
    a(\mathbf{b}_h) = \sum_{\boldsymbol{\omega}}(-1)^{\mathbf{b}_h\cdot\boldsymbol{\omega}}\,F(\boldsymbol{\omega}),
    \qquad
    F(\boldsymbol{\omega}) := \sum_{\mathbf{a}}\overline{B(\mathbf{a},\boldsymbol{\omega})}\,A(\mathbf{a},\boldsymbol{\omega}).
    \label{eq:aWH}
\end{equation}
The amplitude is the Walsh-Hadamard transform of $F$ over the hidden register.

A Walsh-Hadamard transform preserves the sum of squared magnitudes. Applying this through the orthogonality $\sum_{\mathbf{b}_h}(-1)^{\mathbf{b}_h\cdot(\boldsymbol{\omega}-\boldsymbol{\omega}')}=2^m\delta_{\boldsymbol{\omega},\boldsymbol{\omega}'}$ to Eq.~\eqref{eq:aWH},
\begin{equation}
    \mathbb{P}_{\boldsymbol{\theta}}(0^n\mid x) = \sum_{\mathbf{b}_h}|a(\mathbf{b}_h)|^2 = 2^m\sum_{\boldsymbol{\omega}}|F(\boldsymbol{\omega})|^2.
    \label{eq:Pparseval}
\end{equation}
Using $1/\overline{B}=2^N B$ from Eq.~\eqref{eq:Bident},
\begin{equation}
    g(\boldsymbol{\omega}) = 2^{-n}\sum_{\mathbf{a}}\frac{\overline{A}}{\overline{B}} = 2^{m}\sum_{\mathbf{a}}\overline{A(\mathbf{a},\boldsymbol{\omega})}\,B(\mathbf{a},\boldsymbol{\omega}) = 2^{m}\,\overline{F(\boldsymbol{\omega})},
\end{equation}
so $|F(\boldsymbol{\omega})|^2=2^{-2m}|g(\boldsymbol{\omega})|^2$. Substituting into Eq.~\eqref{eq:Pparseval},
\begin{equation}
    \mathbb{P}_{\boldsymbol{\theta}}(0^n\mid x) = \frac{1}{2^m}\sum_{\boldsymbol{\omega}}|g(\boldsymbol{\omega})|^2 = \mathbb{E}_{\boldsymbol{\omega}\sim\mathrm{Unif}(\{0,1\}^m)}\bigl[|g(\boldsymbol{\omega})|^2\bigr].
\end{equation}
\end{proof}

\subsection{Toy example}
\label{app:example}

We illustrate the estimator on a small instance that keeps both the input $x$ and the parameters explicit, two visible qubits coupled to one hidden qubit. We take $n=2$, $m=1$, and encoding precision $d=2$. The data enters the data-layer single-qubit angles additively, $\theta^1_v \to \theta^1_v + \pi x/2^{\,v}$, while $\theta^1_v$, the couplings $\theta^1_{vh}$, and the output-layer angles $\theta^2_v,\theta^2_{vh}$ are trainable.

By Lemma~\ref{lem:factorization}, $g$ factorizes over the two visible qubits, $g(\omega)=G_1(\omega)\,G_2(\omega)$ with $G_v(\omega)=\tfrac12[r_v(0;\omega)+r_v(1;\omega)]$. Evaluating $r_v$ from $\mathrm{LS}_A$ and $\mathrm{LS}_B$ and simplifying gives, for one hidden neighbour,
\begin{equation*}
    |G_v(\omega)|^2 = \tfrac12\Bigl(1-\sin 2\bigl(\theta^1_v+(-1)^\omega\theta^1_{vh}\bigr)\,\sin 2\bigl(\theta^2_v+(-1)^\omega\theta^2_{vh}\bigr)\Bigr).
\end{equation*}
The marginal averages $|g(\omega)|^2=|G_1(\omega)|^2|G_2(\omega)|^2$ over the hidden qubit,
\begin{equation*}
    \mathbb{P}_{\boldsymbol\theta}(0^2\mid x) = \tfrac12\sum_{\omega\in\{0,1\}} |G_1(\omega)|^2\,|G_2(\omega)|^2 .
\end{equation*}
As a concrete check, take $x=0$ with $\theta^1_v=\theta^2_v=\pi/8$ and $\theta^1_{vh}=\theta^2_{vh}=\pi/16$ for both qubits. Then $|G_v(0)|^2=\tfrac12\sin^2(\pi/8)$ and $|G_v(1)|^2=\tfrac12\cos^2(\pi/8)$, so
\begin{equation*}
    \mathbb{P}_{\boldsymbol\theta}(0^2\mid 0)
    = \tfrac18\bigl(\sin^4\tfrac{\pi}{8}+\cos^4\tfrac{\pi}{8}\bigr)
    = \frac{3}{32}.
\end{equation*}
If the two visible qubits were independent, the all-zero marginal would factor into the product of the single-qubit marginals, $\tfrac14\cdot\tfrac14=\tfrac1{16}$. The value $3/32\neq1/16$ shows that the shared hidden qubit correlates them, in other words each factor $|G_v(\omega)|^2$ depends on its own qubit alone, yet the sum over the shared hidden bit $\omega$ ties $G_1$ and $G_2$ together. This correlation is absent for circuits without a shared hidden neighbour.

\newpage
\section{Additional details of numerical results}
\label{app:numerics}

This appendix provides implementation details and additional numerical evidence supporting Sec.~\ref{sec:numerics}.

We first summarize the computational environment and timing conventions. Then, we provide expanded estimator benchmarks, including uncertainty decompositions and memory measurements that were omitted from the main text for readability. Finally, further subsections report training curves, baseline settings, and hardware-execution diagnostics.

\subsection{Hardware and software information}
\label{app:hardware-software}

All CPU statevector benchmarks were performed on a dual-socket Intel Xeon Silver 4216 system with 32 physical cores and 64 logical threads.
Each processor has a base clock frequency of $2.10\,\mathrm{GHz}$.
Exact statevector references were computed using Qiskit~\texttt{version 2.4.0}~\cite{qiskit2024}, with 64 CPU threads.

The Parseval estimator, training loops, and classical scaling experiments were implemented in Python using JAX~\texttt{version 0.8.2}.
GPU runs were executed on NVIDIA Tesla V100S GPUs with 32 GB memory and CUDA~\texttt{version 12.8}.
Each benchmark or training run used a single GPU.
Multiple GPUs were used only to parallelize independent runs, never to distribute a single estimator or training instance.
Reported JAX timings were measured after warm-up, exclude one-time compilation costs, and were synchronized using \texttt{block\_until\_ready()}.
Unless otherwise stated, memory values correspond to a single process on a single GPU.

We used Qiskit for simulated sampling from trained circuits.
When exact probabilities were required, they were obtained using the statevector simulator.
Quantum hardware sampling was performed on \texttt{ibm\_aachen}.
The hardware results use raw measurement statistics and do not apply measurement-error mitigation, zero-noise extrapolation, dynamical decoupling, twirling, postselection, or other mitigation procedures.
Hardware-executed circuits use the backend support for fractional gates.

\subsection{Estimator accuracy and resource scaling}
\label{app:estimator-numerics}
This appendix gives additional details for the estimator benchmarks reported in Sec.~\ref{sec:numerics}, focusing on uncertainty decomposition in the statevector comparison and on the memory model used for the large-scale scaling estimates.

Fig.~\ref{fig:simulation-quality-detailed} gives a more detailed version of the estimator benchmark shown in Fig.~\ref{fig:estimator-numerics}.
The figure separates two sources of uncertainty.
For each total system size $N=n+m$ and Monte Carlo sample count $K$, we use $100$ independently initialized training circuits.
For each parameter seed, the same target probability is estimated $100$ times with independent Monte Carlo samples and compared with an exact statevector reference.

The absolute error for parameter seed $s$ and Monte Carlo repetition $r$ can be introduced as
\begin{equation}
e_{s,r}(N,K)=\bigl|\hat p_{s,r}(N,K)-p^{\mathrm{ref}}_{s}(N)\bigr|.
\end{equation}
We define the per-seed mean error
\begin{equation}
\bar e_s(N,K)=\frac{1}{100}\sum_{r=1}^{100}e_{s,r}(N,K).
\end{equation}
The top row of Fig.~\ref{fig:simulation-quality-detailed} reports
\begin{equation}
\langle \bar e_s(N,K)\rangle_s
\pm
\operatorname{SD}_s\!\left[\bar e_s(N,K)\right],
\end{equation}
which treats each random circuit instance as one experimental unit.
The middle row reports the same central value, but with error bars
\begin{equation}
\left\langle
\operatorname{SD}_r\!\left[e_{s,r}(N,K)\right]
\right\rangle_s ,
\end{equation}
which quantify the Monte Carlo variation for a typical circuit instance.

The two rows therefore answer different questions.
The first shows how much the mean estimator accuracy changes across random training circuits.
The second shows how much a repeated Monte Carlo estimate fluctuates for a fixed circuit.
In this benchmark, the Monte Carlo variation is larger than the across-seed variation, and it decreases with increasing $K$, consistent with the expected $1/\sqrt{K}$ behaviour of Monte Carlo estimates.
The lower row shows the corresponding per-run wall time.
The evaluation remains below $0.5\,\mathrm{ms}$ for almost all tested configurations and reaches about $1.1\,\mathrm{ms}$ only for $K=10^4$ near the largest statevector-accessible system sizes.

\begin{figure}[!h]
    \centering
    \includegraphics[width=\linewidth]{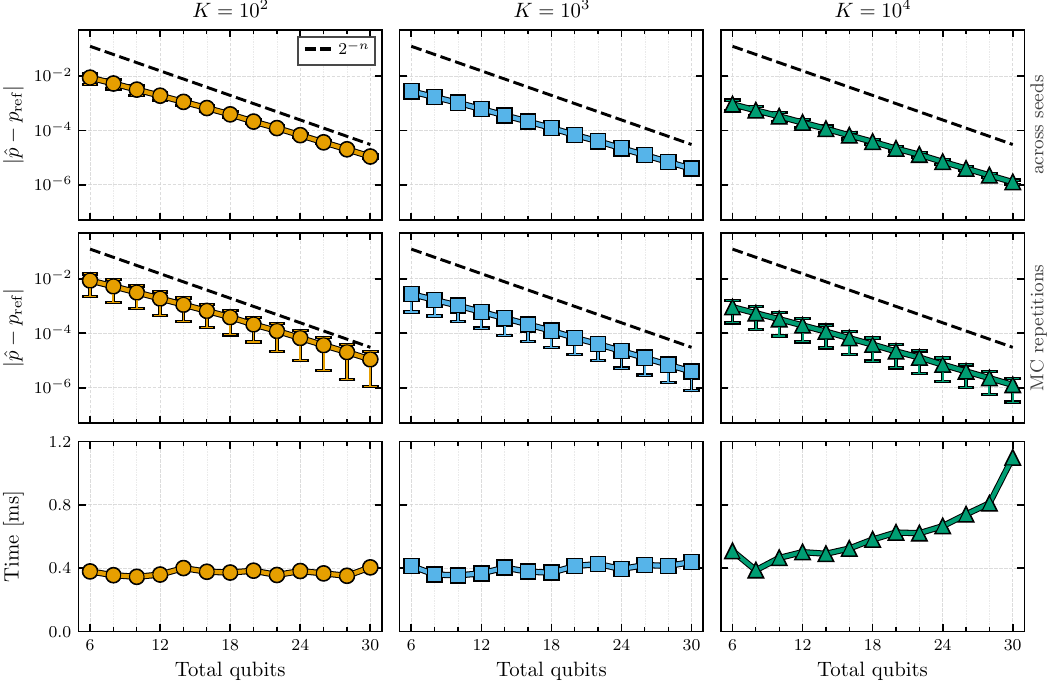}
    \caption{\textbf{Detailed Monte Carlo estimator accuracy and runtime benchmark.}
    The three columns correspond to $K=10^2$, $K=10^3$, and $K=10^4$ Monte Carlo samples.
    The horizontal axis shows the total number of qubits $N=n+m$, with $n=m$ in all panels.
    Each $(N,K)$ point uses $100$ independently initialized training circuits.
    For each circuit, the same visible marginal probability is estimated $100$ times with independent Monte Carlo samples and compared with an exact statevector reference.
    The top row shows the mean absolute error $\langle |\hat p-p_{\mathrm{ref}}| \rangle$ with error bars given by the standard deviation across parameter seeds of the per-seed mean error.
    The middle row shows the same central value with error bars given by the average, over parameter seeds, of the within-seed standard deviation across Monte Carlo repetitions.
    The dashed line marks the natural probability scale $2^{-n}$, where $n$ is the number of visible qubits.
    The bottom row shows the per-run wall time in milliseconds on a single NVIDIA Tesla V100S GPU.
    No timing error bars are shown, since one timing value is logged for each configuration.}
    \label{fig:simulation-quality-detailed}
\end{figure}

The memory curves in Fig.~\ref{fig:estimator-numerics}\textcolor{red}{b} use a conservative estimate rather than the raw measured peak memory.
We use this estimate because the measured resident memory depends on backend-level details of JAX/XLA fusion and allocation, which makes the raw values less stable across software versions and hardware configurations.

For the fully connected visible-hidden bipartite setting, the dominant storage term in the Parseval Monte Carlo contraction scales as $Knm$, where $K$ is the Monte Carlo sample count, $n$ is the number of visible qubits, and $m$ is the number of hidden qubits.
We therefore use the empirical upper-bound model
\begin{equation}
M_{\mathrm{est}}(K,n,m)
=
320~\mathrm{MB}
+
\frac{24}{2^{20}}Knm~\mathrm{MB}.
\label{eq:memory-estimate}
\end{equation}
The fixed term accounts for the baseline JAX/CUDA allocation, while the second term captures the leading $O(Knm)$ storage cost.
The coefficient $24$ B per element is deliberately conservative.
In the scaling plots we set $n=m$, so the total number of qubits is $n+m$.

Fig.~\ref{fig:runtime-scaling-detailed} compares this estimate with measured peak memory.
For $K\leq10^3$, the measured peak remains close to the baseline and grows much more slowly than the estimate.
We attribute this to XLA fusion, which can stream the dominant contraction rather than materializing the full intermediate.
For $K=10^4$, the measured values follow the predicted scaling more closely, and the estimate correctly predicts the device-memory limit.
We use Eq.~\eqref{eq:memory-estimate} as a conservative scheduling rule.
It overestimates memory in the fusion-dominated regime, but it is sufficiently tight in the large-$K$, large-$n,m$ regime that determines feasibility.

\begin{figure}[!h]
    \centering
    \includegraphics[width=\linewidth]{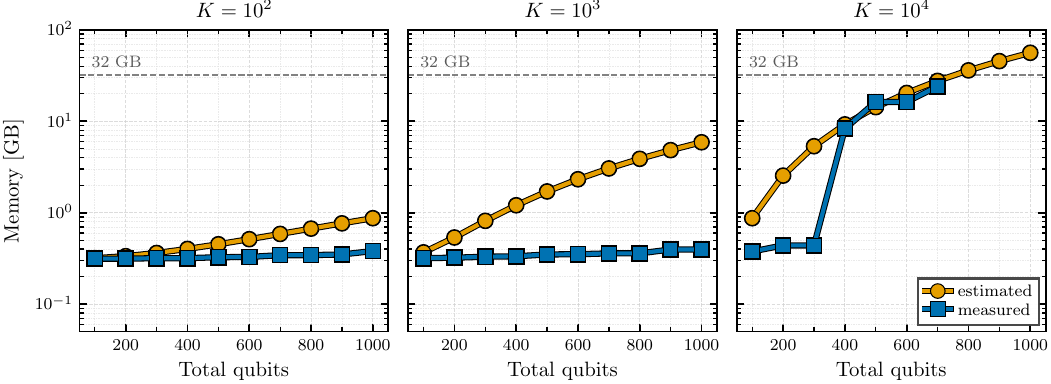}
    \caption{\textbf{Peak GPU memory for the Parseval Monte Carlo simulator.}
    Measured and estimated peak memory are shown as a function of the total number of qubits for $K=10^2$, $K=10^3$, and $K=10^4$ Monte Carlo samples.
    Measurements are per-process maxima obtained by polling \texttt{nvidia-smi} during execution on a single NVIDIA Tesla V100S GPU with 32 GB memory.
    The analytic estimate is given by Eq.~\eqref{eq:memory-estimate}.
    The dashed horizontal line marks the 32 GB device limit.
    In these benchmarks we set $n=m$, where $n$ is the number of visible qubits and $m$ is the number of hidden qubits.
    Missing measured points for $K=10^4$ correspond to runs that exceeded device memory.
    The estimate is conservative for $K\leq10^3$, where XLA fusion appears to avoid materializing the dominant intermediate, and becomes tighter in the large-$K$ regime that determines feasibility.}
    \label{fig:runtime-scaling-detailed}
\end{figure}

\subsection{Univariate benchmark training details}
\label{app:univariate-details}

Fig.~\ref{fig:exp1-learning-curves} shows the training dynamics for the four univariate benchmarks reported in Fig.~\ref{fig:univarite-numerics}\textcolor{red}{a}.
The columns correspond to the Gaussian, Gaussian mixture, cosine, and L\'evy targets.
The rows show the mini-batch negative log-likelihood, the full-distribution KL divergence $\KL(p_{\mathrm{target}}\|p_{\mathrm{model}})$, and the total variation distance $\TVD(p_{\mathrm{target}},p_{\mathrm{model}})$.
The NLL is logged every $10$ optimization steps on mini-batches of size $64$, which explains the visible fluctuations in the first row.
KL and TVD are evaluated every $1000$ steps from the full target and model PMFs, and therefore appear as checkpointed curves.

All four runs use the same architecture and hyperparameters, differing only in the target distribution.
The model has $m=4$ hidden qubits with fully connected visible-hidden bipartite connectivity, and gradients are estimated with $K=10^3$ hidden-register samples per optimization step.
We train with Adam for $10^4$ steps using batch size $64$, initial learning rate $0.15$, and exponential decay $0.99$ every $10$ steps.

The target distributions are represented on the \texttt{float16} grid over $x\in[0,1]$, corresponding to an effective $2^{16}$-point support.
Since the DQGM is formulated in a Fourier representation, we use one additional encoding qubit per feature as a frequency buffer near the Nyquist limit.
This avoids using the highest representable Fourier component in the readout, where discretization and aliasing effects are most pronounced.
In practice, using a single buffer qubit improves stability, while increasing the buffer beyond one qubit did not lead to further improvement in our experiments.
We therefore use this convention throughout, encoding each feature with one additional qubit and omitting the highest-frequency readout bit when forming the evaluated model distribution.

The training curves show the same qualitative behavior across the four targets.
The NLL decreases rapidly and then fluctuates around a mini-batch noise floor, while KL and TVD decrease during the early part of training and then plateau.
The cosine target has the largest residual error, consistent with the metrics reported in the main text.
Each univariate run completes in less than five minutes on a single GPU.

\begin{figure}[!h]
    \centering
    \includegraphics[width=\linewidth]{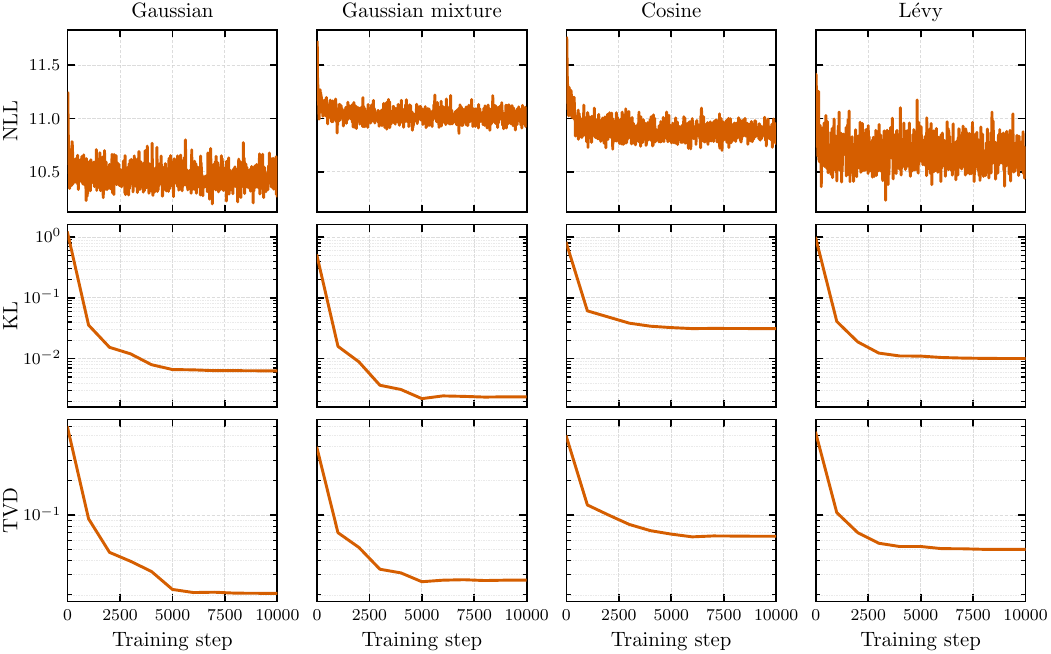}
    \caption{\textbf{Learning curves for the univariate benchmark distributions.}
    Columns correspond to the Gaussian, Gaussian mixture, cosine, and L\'evy targets used in Fig.~\ref{fig:univarite-numerics}\textcolor{red}{a}.
    The rows show the mini-batch negative log-likelihood, the full-distribution KL divergence $\KL(p_{\mathrm{target}}\|p_{\mathrm{model}})$, and the total variation distance $\TVD(p_{\mathrm{target}},p_{\mathrm{model}})$.
    NLL is logged every $10$ optimization steps on mini-batches of size $64$.
    KL and TVD are computed every $1000$ steps using the full target and model PMFs.
    All runs use the same architecture, initialization, and hyperparameters with $m=4$ hidden qubits and $K=10^3$ hidden-register samples per gradient estimate.}
    \label{fig:exp1-learning-curves}
\end{figure}

\subsection{Hardware run details}
\label{app:hardware}

We detail the steps taken to execute the trained DQGM circuits on the device.
Three approximations separate the trained circuit from the one we run, and we describe each in turn.

We map each circuit to the device coupling map with the SABRE transpiler in Qiskit.
SABRE is a stochastic transpiler, so repeated runs return
circuits of differing depth and two-qubit gate count.
We therefore transpile each circuit $100$ times and retain the realization with the fewest two-qubit gates.

We then truncate the parametrized gates.
Many trained rotation angles are small and contribute little to the output state, so we drop every parametrized gate whose angle satisfies $|\theta| < 0.05$.
This reduces the gate count at a controlled cost in accuracy.

We replace the exact inverse quantum Fourier transform (QFT) with an approximate inverse QFT, omitting the controlled-phase rotations between qubits separated by more than three positions.
The dropped rotations carry the smallest phases, so the resulting degree-$3$ approximation reproduces the exact transform to within a small error while removing a large fraction of the two-qubit gates.

Figure~\ref{fig:exp4-detailed} traces the effect of these approximations across the pipeline.
We run the resulting circuits on the IBM backend \texttt{ibm\_aachen} with $10^5$ shots and it uses 30 seconds hardware time.
The approximate circuit (third column) tracks the ideal simulation (second columnd) closely, so the dominant loss of fidelity arises at the hardware stage (fourth column).

\begin{figure}[h]
    \centering
    \includegraphics[width=\textwidth]{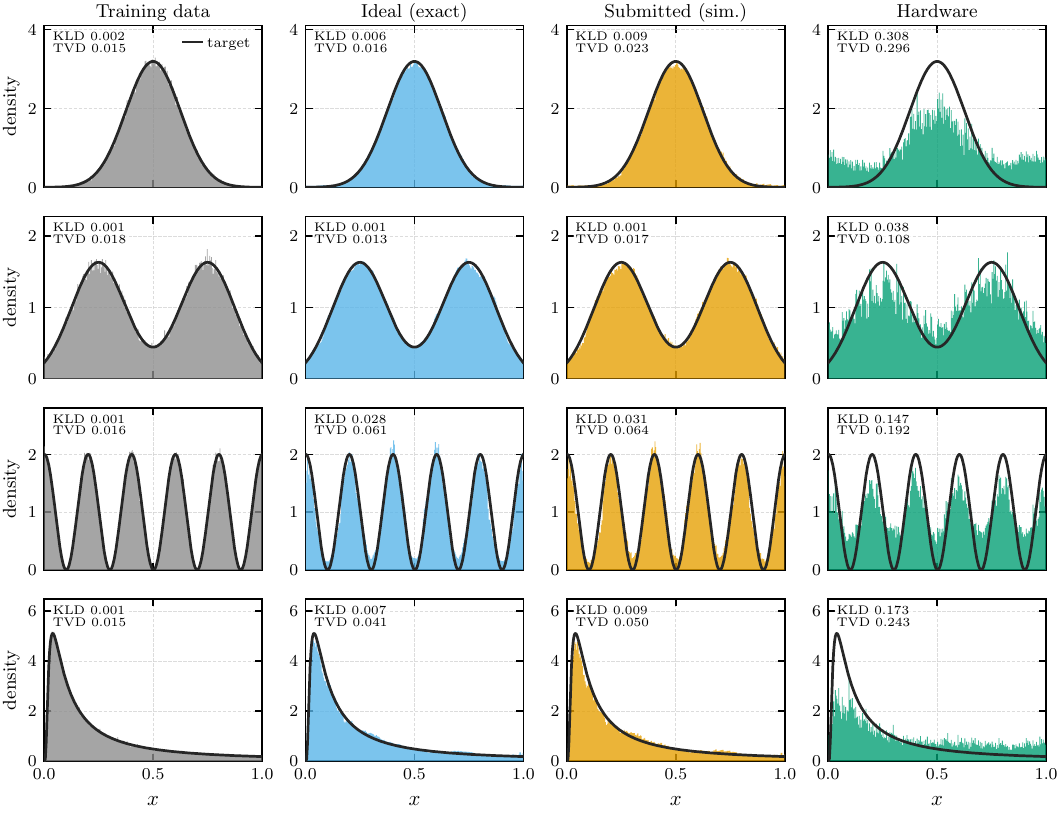}
    \caption{\textbf{Stage-by-stage comparison of the learned univariate densities on hardware.} 
        Each row is a separate target distribution (top to
        bottom: Gaussian, Gaussian mixture, cosine, and L\'evy), and the columns
        trace the pipeline from data to device, left to right.
        \textbf{(i)}~\emph{Training data}, the empirical histogram of the samples
        the model was trained on.
        \textbf{(ii)}~\emph{Ideal (exact)}, the trained DQGM circuit simulated
        exactly with a standard inverse QFT and no gate truncation.
        \textbf{(iii)}~\emph{Submitted (sim.)}, the exact statevector of the
        circuit actually submitted to the device, with the approximate QFT and
        gate truncation applied, isolating the approximation error from device
        noise.
        \textbf{(iv)}~\emph{Hardware}, the counts returned by the device.
        In every panel the solid black line is the analytical target density.
        Each panel is annotated with the Kullback--Leibler divergence (KLD) and
        total variation distance (TVD) to that target. For the training-data
        column these quantify the finite-size training set.}
    \label{fig:exp4-detailed}
\end{figure}

\subsection{Bivariate benchmarks --- Part 1}
\label{app:2d-iqp}

We benchmark DQGM against an instantaneous quantum polynomial (IQP) circuit on
the bivariate targets. We train the IQP circuit with the maximum mean discrepancy
(MMD) loss and DQGM with the negative log-likelihood (NLL) loss, and for each model
we compare the exact loss estimator against a Monte Carlo (MC) estimator.

Figure~\ref{fig:exp6} compares two IQP variants on the four-peak Gaussian target.
The top row shows the target distribution, the best IQP trained with the exact MMD
loss, and the best IQP trained with the same loss estimated by Monte Carlo sampling.
The exact-trained IQP reproduces the target most closely, but exact training does
not scale, so we include it only to gauge the expressibility of the circuit. The
Monte Carlo-trained IQP recovers the low-order Walsh--Hadamard components shown in
the bottom row, while the high-order components remain inaccurate.

Table~\ref{tab:sweep} reports the final metrics across a hyperparameter sweep. We
vary the maximum Pauli weight, the number of hidden qubits, and the estimator, and
we record the resulting gate count together with the final KL and TVD to the
target. Each entry is the mean and standard deviation over seeds.

\begin{figure}[!t]
    \centering
    \includegraphics[width=\linewidth]{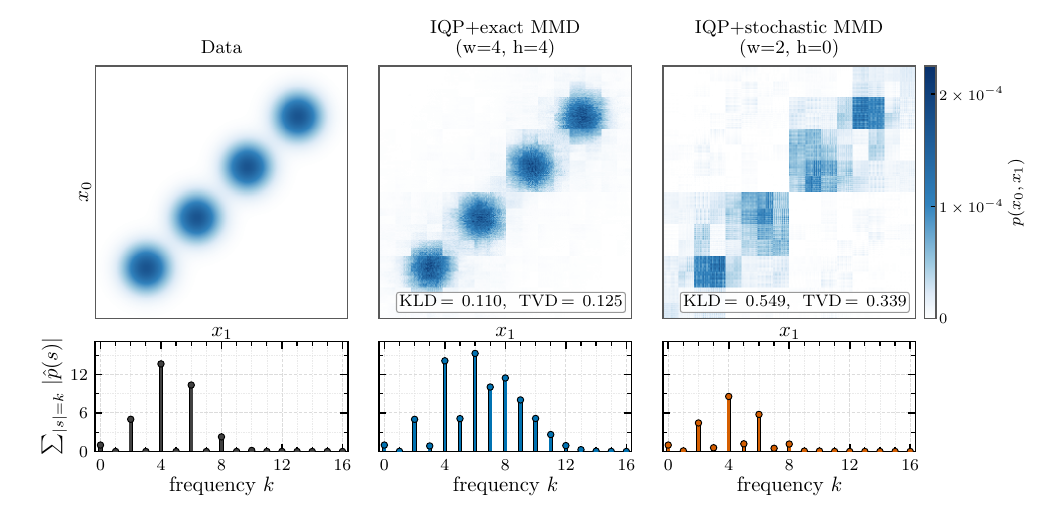}
    \caption{\textbf{Bivariate benchmark of exact and Monte Carlo IQP training.}
    Comparison on a four-peak Gaussian target with an $8$-bit discretization for
    each coordinate. The top row shows the target distribution (left), the best
    IQP trained with the exact MMD loss (middle), and the best IQP+MMD trained with
    a Monte Carlo (MC) estimator (right). Exact training does not scale to large
    numbers of qubits, and we include it as a sanity check on the expressibility of
    the IQP circuit. The IQP circuit contains all one-qubit $\mathrm{RZ}$ and all two-qubit
    $\mathrm{RZZ}$ rotations on the visible and hidden qubits ($m=4$), trained with a
    multi-kernel MMD objective using bandwidths $\sigma\in\{0.01,0.1,1.0\}$. The
    bottom row shows the total absolute Walsh--Hadamard spectral weight grouped by
    Hamming weight, $\sum_{|s|=k}|\hat p(s)|$. The low-order components match the
    target for both training methods, while the high-order components deviate. For
    both models we plot the best-performing seed. Quantitative metrics averaged over
    multiple seeds and the training details are reported in
    Table~\ref{tab:sweep}.}
    \label{fig:exp6}
\end{figure}
\begin{table}[!t]
  \centering
  \caption{\textbf{Hyperparameter sweep for the bivariate benchmark.} Final
    Kullback--Leibler divergence (KL) and total variation distance (TVD) to the
    target for each configuration, given as mean $\pm$ standard deviation over
    seeds. We sweep the maximum Pauli weight, the number of hidden qubits, and the
    estimator (exact or Monte Carlo, MC), and we report the gate count.}
  \label{tab:sweep}
  \begin{tabular*}{\textwidth}{@{\extracolsep{\fill}}lllrrrrr@{}}
    \toprule
    Model & Loss & Estimator & Max.\ weight & \#Hidden & \#Gates & Final KL & Final TVD \\
    \midrule
    IQP & MMD & exact & 2 & 0 & 136 & $0.767 \pm 0.000$ & $0.435 \pm 0.000$ \\
    IQP & MMD & MC & 2 & 0 & 136 & $0.626 \pm 0.161$ & $0.390 \pm 0.052$ \\
    IQP & MMD & exact & 2 & 2 & 171 & $0.692 \pm 0.000$ & $0.404 \pm 0.000$ \\
    IQP & MMD & MC & 2 & 2 & 171 & $0.592 \pm 0.049$ & $0.398 \pm 0.029$ \\
    IQP & MMD & exact & 2 & 4 & 210 & $0.532 \pm 0.000$ & $0.357 \pm 0.000$ \\
    IQP & MMD & MC & 2 & 4 & 210 & $0.581 \pm 0.018$ & $0.389 \pm 0.008$ \\
    IQP & MMD & exact & 3 & 0 & 696 & $0.441 \pm 0.128$ & $0.294 \pm 0.053$ \\
    IQP & MMD & MC & 3 & 0 & 696 & $1.043 \pm 0.025$ & $0.501 \pm 0.008$ \\
    IQP & MMD & exact & 3 & 2 & 987 & $0.312 \pm 0.046$ & $0.257 \pm 0.018$ \\
    IQP & MMD & MC & 3 & 2 & 987 & $0.656 \pm 0.014$ & $0.430 \pm 0.007$ \\
    IQP & MMD & exact & 3 & 4 & 1350 & $0.317 \pm 0.040$ & $0.258 \pm 0.024$ \\
    IQP & MMD & MC & 3 & 4 & 1350 & $0.590 \pm 0.029$ & $0.416 \pm 0.012$ \\
    IQP & MMD & exact & 4 & 0 & 2516 & $0.225 \pm 0.026$ & $0.188 \pm 0.014$ \\
    IQP & MMD & MC & 4 & 0 & 2516 & $1.285 \pm 0.034$ & $0.574 \pm 0.010$ \\
    IQP & MMD & exact & 4 & 2 & 4047 & $0.121 \pm 0.030$ & $0.142 \pm 0.025$ \\
    IQP & MMD & MC & 4 & 2 & 4047 & $1.556 \pm 0.004$ & $0.709 \pm 0.001$ \\
    IQP & MMD & exact & 4 & 4 & 6195 & $0.121 \pm 0.011$ & $0.140 \pm 0.011$ \\
    IQP & MMD & MC & 4 & 4 & 6195 & $1.478 \pm 0.002$ & $0.704 \pm 0.000$ \\
    \midrule
    DQGM (ours) & NLL & MC & 2 & 1 & 72 & $0.699 \pm 0.001$ & $0.507 \pm 0.001$ \\
    DQGM (ours) & NLL & MC & 2 & 2 & 108 & $0.169 \pm 0.001$ & $0.178 \pm 0.002$ \\
    DQGM (ours) & NLL & MC & 2 & 3 & 144 & $0.281 \pm 0.197$ & $0.252 \pm 0.136$ \\
    DQGM (ours) & NLL & MC & 2 & 4 & 180 & $0.203 \pm 0.172$ & $0.190 \pm 0.121$ \\
    \bottomrule
  \end{tabular*}
\end{table}

\newpage
\subsection{Bivariate bencmarks --- Part 2}
\label{app:2d-classical}

We tune each classical baseline under the same protocol as the quantum model, so
that the comparison reflects the models rather than unequal tuning effort. We train
with Adam optimized for a fixed budget of $10^4$ steps under an exponential learning-rate decay with per-step factor $\gamma = r^{1/\tau}$, where $r$ is the decay rate and $\tau$ the
number of transition steps ($r=1$ gives a constant learning rate). We train all models
on the \emph{blobs-on-spiral} dataset ($10^5$ samples, $2$ features at $8$-bit
precision).

We select hyperparameters with the Tree-structured Parzen Estimator (TPE)
sampler~\cite{optuna_2019}, using $10$ random startup trials and a median pruner. The
pruner becomes active after $12$ completed trials and prunes a trial once it has
reported for at least two seeds. We organize the search as a capacity grid. For each
fixed model size we run an independent study of $30$ trials over the optimization
hyperparameters of Table~\ref{tab:hpo-space}, scoring each trial by the agreement
between the generated and target distributions over $3$ seeds. We retrain the selected
configurations with $10$ random seeds and draw the run shown in the figures from these.
We show three capacities, denoted \textbf{a}, \textbf{b}, and \textbf{c}.

\begin{table}[!h]
\centering
\caption{\textbf{Optimization hyperparameters swept during the search}, sampled per
capacity point. For the diffusion baseline we additionally sweep the number of
denoising steps $T \in \{50,\,100,\,200\}$ and the terminal noise level
$\beta_{\text{end}} \in \{0.02,\,0.05\}$, using a linear schedule with
$\beta_{\text{start}}=10^{-4}$.}
\label{tab:hpo-space}
\begin{tabular}{lll}
\toprule
Hyperparameter & Range / set & Sampling \\
\midrule
Learning rate        & $[10^{-4},\,5\times10^{-3}]$ & log-uniform \\
Batch size           & $\{256,\,512\}$              & categorical \\
LR decay rate $r$    & $\{0.99,\,0.995,\,1.0\}$     & categorical \\
LR transition $\tau$ & $\{500,\,1000,\,2000\}$      & categorical \\
\bottomrule
\end{tabular}
\end{table}

We compare against two classical generative baselines, a normalizing flow and a
diffusion model. The normalizing flow is a neural spline flow~\cite{nsf} with $K$
rational-quadratic coupling transforms. The conditioner of each transform is an MLP
with the listed hidden widths, and each spline uses $B$ bins. We implement it with
\texttt{Zuko}~\cite{zuko} and \texttt{PyTorch}~\cite{pytorch}.
Table~\ref{tab:nf-arch} lists the three plotted capacities.

\begin{table}[h]
\centering
\caption{\textbf{Normalizing-flow architectures for the plotted capacities.}}
\label{tab:nf-arch}
\begin{tabular}{lcccr}
\toprule
Model & Coupling transforms $K$ & Conditioner hidden widths & Spline bins $B$ & \#Params \\
\midrule
a & 1 & $[16]$       & 4 & 422 \\
b & 2 & $[32]$       & 8 & 3{,}228 \\
c & 2 & $[32,\,64]$  & 8 & 10{,}396 \\
\bottomrule
\end{tabular}
\end{table}

The diffusion baseline is a denoising diffusion probabilistic model (DDPM)~\cite{ddpm}
whose noise-prediction network $\epsilon_\theta$ is an MLP of hidden width $H$ and
depth $L$, conditioned on a sinusoidal time embedding of dimension $d_t$ projected
through a two-layer MLP. We tune the number of denoising steps and the noise schedule
jointly with the optimization hyperparameters, and neither affects the parameter
count. Table~\ref{tab:ddpm-arch} lists the three plotted capacities.

\begin{table}[h]
\centering
\caption{\textbf{Diffusion (DDPM) architectures for the plotted capacities.}}
\label{tab:ddpm-arch}
\begin{tabular}{lcccr}
\toprule
Model & Hidden width $H$ & Depth $L$ & Time-embed dim $d_t$ & \#Params \\
\midrule
a & 32  & 2 & 32 & 4{,}354 \\
b & 64  & 3 & 64 & 21{,}058 \\
c & 128 & 4 & 64 & 91{,}394 \\
\bottomrule
\end{tabular}
\end{table}

Figure~\ref{fig:exp3-baselines} shows the distributions learned by the two classical
baselines across the three capacities. Both the normalizing flow (top row) and the
diffusion model (bottom row) reproduce the \emph{blobs-on-spiral} structure more
closely as capacity increases, with the largest capacity (c) resolving the individual
modes better that the smallest capacity (a) blurs together. 
However, even at the number of parameters comparable to the training set size (100,000) none of the model can not resolve all modes completely.
In the main text we show, for each model, the single best run across all capacities and seeds.
Here we display the full per-capacity sweep so that the dependence on model size is visible.

\begin{figure}[h]
\centering
\includegraphics[width=\linewidth]{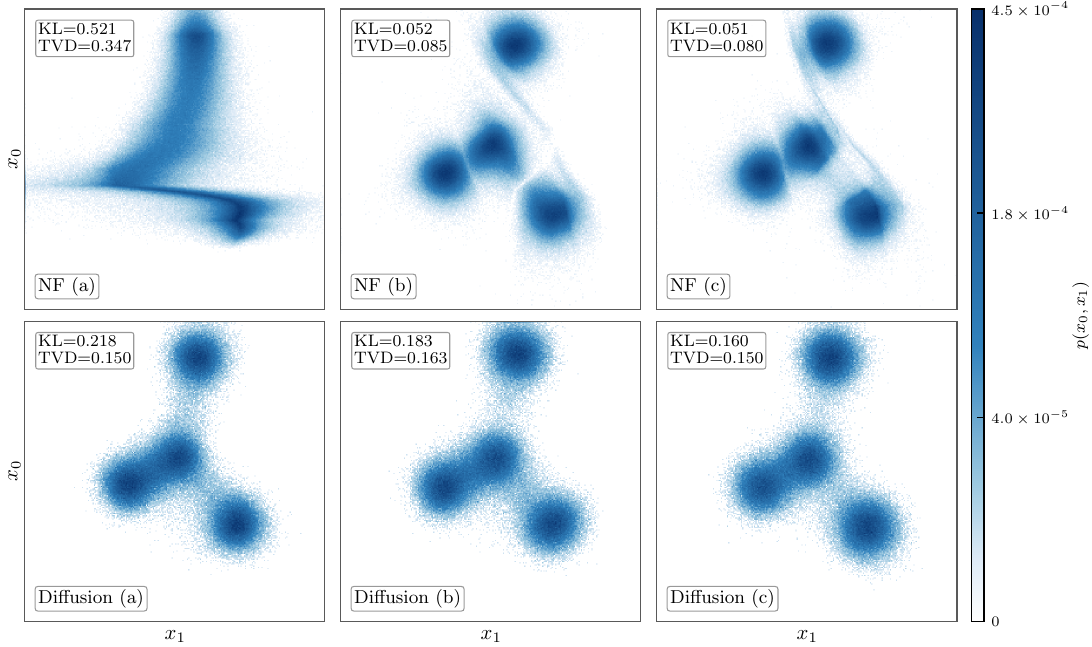}
\caption{%
\textbf{Classical baseline samples on \emph{blobs-on-spiral}.}
Learned joint probability mass over the $2^{8}\times 2^{8}$ grid for the
normalizing-flow (top row) and diffusion (bottom row) baselines at three model
capacities, denoted (a), (b), and (c) in order of increasing size. Architectures and
parameter counts are listed in Tables~\ref{tab:nf-arch} and~\ref{tab:ddpm-arch}.
Each panel shows the best of 10 seeds for the configuration selected by
the hyperparameter search of App.~\ref{app:2d-classical}. All panels share a common color scale with power-law scaling, where darker shading
indicates higher probability mass. The target distribution is shown in
Fig.~\ref{fig:iqp-benchmark}.
}
\label{fig:exp3-baselines}
\end{figure}

\subsection{Multivarite benchmark}

The finance dataset consists of daily log-returns $r_t = \log(c_t / c_{t-1})$ of four
assets over the 2008 calendar year, namely the S\&P~500 (\texttt{\^{}GSPC}), the
Nikkei~225 (\texttt{\^{}N225}), a gold ETF (\texttt{GLD}), and a long-dated U.S.\
Treasury bond ETF (\texttt{TLT}), giving $k = 4$ features. We auto-adjust the closing
prices, align them across the four assets, dropping any day missing a quote, for
example from differing market holidays, and difference the aligned series. This leaves
$N = 236$ daily returns for 2008, and we use the entire series. We scale and discretise
each feature with $d = 4$ bits ($2^4 = 16$ levels per feature), for a total of
$k\,d = 16$ visible qubits and a joint support of $16^4 = 65{,}536$ bins. The small
sample size relative to the support makes this a deliberately data-scarce, heavy-tailed
benchmark.

For this benchmark the IQP+MMD baseline uses maximum Pauli weight $2$ and $m = 4$
hidden qubits. As shown in Fig.~\ref{fig:exp5-finance-heatmaps}, the DQGM reproduces the pairwise
marginals of the finance data, recovering the shape and orientation of most feature
pairs, whereas the IQP+MMD baseline misses the correlation structure of nearly every
pair and attains higher KLD and TVD throughout.

\begin{figure}[!h]
\centering
\includegraphics[width=\textwidth]{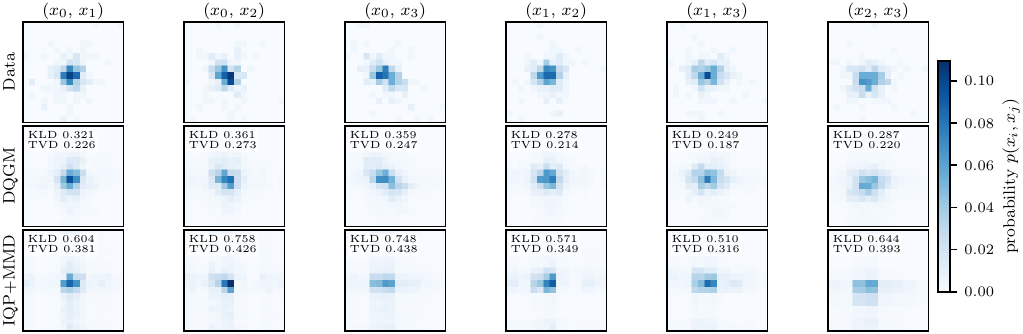}\\[2pt]
\caption{\textbf{Pairwise two-dimensional marginals $p(x_i, x_j)$ for the multivariate
        benchmark on the finance dataset.} The top row is the training-data (target)
        distribution and the lower rows are the trained DQGM and IQP+MMD models, and
        the columns are the unordered feature pairs. All panels share a
        common color scale (white at zero) and a single color bar, and each model
        panel is annotated with the Kullback--Leibler divergence (KLD) and total
        variation distance (TVD) to the matching data panel.}
\label{fig:exp5-finance-heatmaps}
\end{figure}

\clearpage
\newpage
\section{Supplementary numerical results}
\label{app:supp-numerics}

This appendix collects additional numerical results for different datasets.

\subsection{Additional Multivariate distributions}

Beyond the finance dataset of the main text, we test the model on a
higher-dimensional multivariate distribution drawn from particle physics. We use
the JetNet dataset~\cite{jetnet}, a standard benchmark for generative models of
jet substructure, which gives a heavy-tailed joint distribution over a larger
feature set than the finance data.

We use the JetNet gluon-jet sample (jet type \texttt{g}). For each jet we take the
relative transverse momentum $p_T^{\mathrm{rel}}$ (the fraction of the jet transverse
momentum carried by a constituent) of its five leading (highest-$p_T$) particles,
giving $k = 5$ features. We retain only $p_T^{\mathrm{rel}}$ and discard the angular
constituent coordinates. We scale each feature to the unit interval by min--max
normalisation and discretise it with $d = 3$ bits ($2^3 = 8$ levels per feature), for
a total of $k\,d = 15$ visible qubits and a joint support of $8^5 = 32{,}768$ bins. We
draw $N = 10^{6}$ jets, the full statistics used for training. The IQP+MMD baseline
uses maximum Pauli weight $2$ and $m = 4$ hidden qubits.

Fig.~\ref{fig:exp5-HEP-heatmaps} shows the pairwise two-dimensional marginals and
Fig.~\ref{fig:exp5} the correlation-rank plot. On JetNet the DQGM recovers the ordering of the ten pairwise correlations, with a Spearman rank correlation of $\rho = 0.93$ against
$\rho = 0.38$ for the IQP+MMD baseline, and attains lower KLD and TVD across the pair
marginals. The five-feature support and the coarse $3$-bit precision make this a
harder benchmark than the finance data, yet the DQGM still captures the dominant
cross-feature structure.

\begin{figure}[!h]
    \centering
    \includegraphics[width=.5\linewidth]{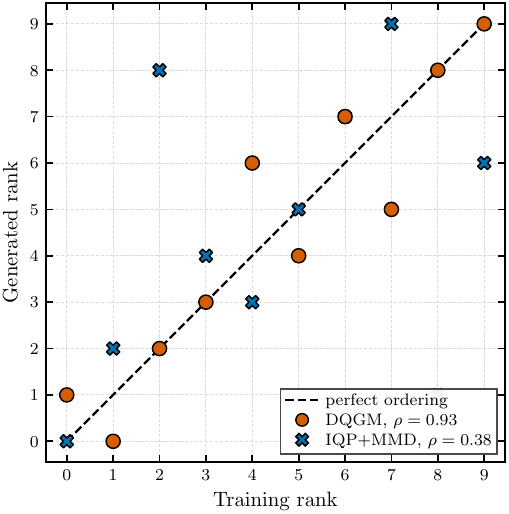}
    \caption{\textbf{Preservation of the pairwise-correlation ordering on JetNet.}
        For every feature pair we rank the Pearson correlations of the training data
        against those of one million samples drawn from each model, so points on the
        dashed diagonal indicate that the model reproduces the empirical ordering
        exactly. The legend reports the Spearman rank correlation $\rho$. The JetNet
        particle-physics dataset has $k=5$ features and $\binom{5}{2}=10$ pairs, and we
        compare DQGM against the IQP+MMD baseline.}
    \label{fig:exp5}
\end{figure}

\begin{figure}[!h]
\centering
\includegraphics[width=\textwidth]{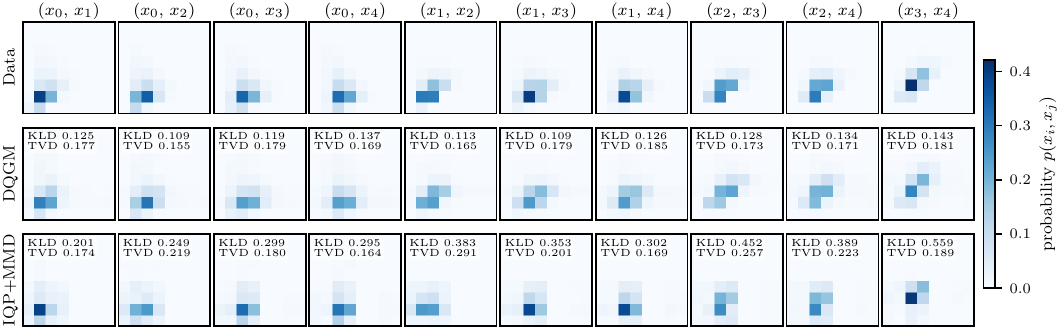}
\caption{\textbf{Pairwise two-dimensional marginals $p(x_i, x_j)$ on JetNet.}
        The top row is the training-data (target) distribution and the lower rows are
        the trained DQGM and IQP+MMD models, and the columns are the unordered feature
        pairs. All panels share a common colour scale (white at zero) and a single colour
        bar, and each model panel is annotated with the Kullback--Leibler divergence
        (KLD) and total variation distance (TVD) to the matching data panel. The JetNet
        particle-physics dataset has $k=5$ features and $\binom{5}{2}=10$ pair columns.}
\label{fig:exp5-HEP-heatmaps}
\end{figure}

\subsection{Upsampling}
\label{app:upsampling}

A trained DQGM represents a continuous density, and the output precision is set by the
number of qubits read out through the per-feature inverse QFT rather than fixed at
training time. We can therefore train at one precision and sample at a higher one
without retraining \cite{kyriienko2024,martinezdelejarza2025}, which we refer to as upsampling. We append ancilla qubits in $\ket{0}$ to each per-feature inverse QFT, refining the readout grid while leaving the learned parameters unchanged.

We train DQGM circuits at $8$-bit output precision and sample from it using the same parameters at $10$ and $12$ bits. Fig.~\ref{fig:exp4-extended} shows the result for two targets, a cosine density and a Gaussian mixture. The KLD and TVD to the analytical target at the corresponding precision are essentially unchanged across the three readouts, so the
upsampled distributions interpolate the learned density onto the finer grid rather than
improve or degrade the fit.

\begin{figure}[!h]
    \centering
    \includegraphics[width=.92\textwidth]{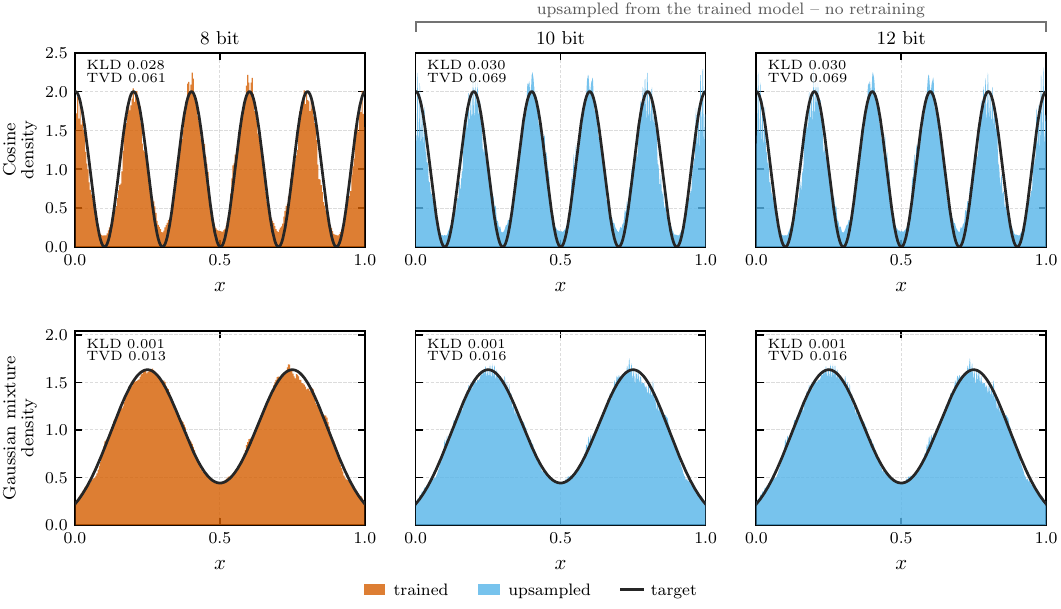}
    \caption{\textbf{Sampling a trained circuit at higher precision by upsampling, with
        no retraining.} Each row is a separate trained DQGM run, a cosine density in the
        top row and a Gaussian mixture in the bottom row, both trained at $8$-bit output
        precision. The leftmost panel shows the trained $8$-bit readout, and the two
        bracketed panels show the same learned parameters read out at $10$ and $12$ bits
        by appending ancilla qubits in $\ket{0}$ at the MSB end of each per-feature
        inverse QFT, with no additional training. All distributions are exact,
        shot-noise-free statevector simulations. The solid black line is the analytical
        target at the corresponding precision, and each panel is annotated with the
        Kullback--Leibler divergence (KLD) and total variation distance (TVD) to it. The
        error is essentially unchanged across precisions, showing that the learned
        continuous density is interpolated onto the finer grid rather.}
    \label{fig:exp4-extended}

\end{figure}

\end{document}